\documentclass[final]{dmtcs-episciences}

\pdfoutput=1
\usepackage{amssymb,amsbsy,amsmath,amsfonts,amssymb,amscd,amsthm}

\usepackage{complexity}
\usepackage{soul}
\usepackage{bbold}
\usepackage{color}

\newcommand{\EnumP}{\mathrm{\mathsf EnumP}}
\newcommand{\IncP}{\mathrm{\mathsf IncP}}
\newcommand{\DelayP}{\mathrm{\mathsf DelayP}}

\newcommand{\cH}{\ensuremath{\mathcal{H}}}
\newcommand{\cF}{\ensuremath{\mathcal{F}}}
\newcommand{\cT}{\ensuremath{\mathcal{T}}}
\newcommand{\cU}{\ensuremath{\mathcal{U}}}
\newcommand{\cE}{\ensuremath{\mathcal{E}}}
\newcommand{\ccS}{\ensuremath{\mathcal{S}}}

\newcommand{\1}{\ensuremath{\mathbb{1}}}
\newcommand{\0}{\ensuremath{\mathbb{0}}}
\newcommand{\Cl}{\ensuremath{Cl}}
\newcommand{\enumSat}[1]{\textsc{Enum#1}}
\newcommand{\EnumKMaxIndSet}{\textsc{EnumKMaxIndSet}}
\newcommand{\EnumClo}{\textsc{EnumClosure}}
\newcommand{\EnumUDClo}{\textsc{EnumUDClosure}}
\newcommand{\EnumCloMin}{\textsc{EnumClosureMin}}
\newcommand{\EnumCloMax}{\textsc{EnumClosureMax}}
\newcommand{\Mem}{\textsc{Membership}}

\newcommand{\ExtClo}{\textsc{ExtClosure}}
\newcommand{\zero}{\ensuremath{\mathbf{0}}}
\newcommand{\one}{\ensuremath{\mathbf{1}}}

\newcommand{\CloMax}{\textsc{ClosureMax}}

\newtheorem{openproblem}{\bf Open problem}

\received{2017-12-12}
\revised{2019-04-20, 2018-09-14}
\accepted{2019-06-5}
\publicationdetails{21}{2019}{3}{22}{4143}

 \newtheorem{theorem}{Theorem}
 \newtheorem{lemma}[theorem]{Lemma}
 \newtheorem{proposition}[theorem]{Proposition}
  \newtheorem*{proposition*}{Proposition}
 \newtheorem{corollary}[theorem]{Corollary}
 \theoremstyle{definition}
 \newtheorem{definition}[theorem]{Definition}

\title{Efficient enumeration of solutions produced by closure operations}
\author{Arnaud Mary\affiliationmark{1}
    \and Yann Strozecki \affiliationmark{2}}

\affiliation{Universit{\'e} Lyon 1 ; CNRS, UMR5558, LBBE / INRIA - ERABLE\\
Universit{\'e} de Versailles Saint-Quentin-en-Yvelines, DAVID laboratory}

\begin{document}
\maketitle

\begin{abstract}
 In this paper we address the problem of generating all elements obtained by the saturation of an initial set by some operations. More precisely, we prove that we can generate the closure of a boolean relation (a set of boolean vectors) by polymorphisms  with a polynomial delay. Therefore we can compute with polynomial delay the closure of a family of sets by any set of ``set operations'': union, intersection, symmetric difference, subsets, supersets $\dots$). To do so, we study the $\textsc{Membership}_{\mathcal{F}}$ problem: for a set of operations $\mathcal{F}$, decide whether an element belongs to the closure by $\mathcal{F}$ of a family of elements. In the boolean case, we prove that $\textsc{Membership}_{\mathcal{F}}$ is in $\P$ for any set of boolean operations $\mathcal{F}$. When the input vectors are over a domain larger than two elements, we prove that the generic enumeration method fails, since $\textsc{Membership}_{\mathcal{F}}$ is $\NP$-hard for some $\mathcal{F}$. We also study the problem of generating minimal or maximal elements of closures and prove that some of them are related to well known enumeration problems such as the enumeration of the circuits of a matroid or the enumeration of maximal independent sets of a hypergraph.
\end{abstract}

\section{Introduction}

An enumeration problem is the task of listing all elements of a set without redundancies.
Since the set to generate may be of exponential cardinality in the size of the input, the complexity of enumeration problems generally are  measured in term of the input size \emph{and} output size. Enumeration algorithms whose complexity depends both on the input and the output are called \emph{output sensitive} and when the dependency is polynomial in the sum of both measures, they are called \emph{output polynomial}. Another more precise notion of complexity, is the \emph{delay} which measures the time between the production of two consecutive solutions. We are especially interested in problems solvable with a  delay polynomial in the input size, which are considered as the \emph{tractable problems} in enumeration complexity. For instance, the spanning trees, the simple cycles~\cite{read1975bounds} or the maximal independent sets~\cite{JohnsonP88} of a graph can be enumerated with polynomial delay.

If we allow the delay to grow during the algorithm, we obtain \emph{polynomial incremental time} algorithms: the first $k$ solutions can be enumerated in a time polynomial in $k$ and in the size of the input. Many problems which can be solved in polynomial incremental time have the following form: given a set of elements and a polynomial time function acting on tuples of elements, produce the closure of the set by the function. For instance the following problems can be solved in polynomial incremental time: the enumeration of the circuits of a matroid~\cite{khachiyan2005complexity} and the enumeration of the vertices of restricted polyhedra~\cite{elbassioni2018enumerating}.

In this article, we try to understand when saturation problems, which by definition can be solved in polynomial incremental time, can be in fact solved by a polynomial delay algorithm. We would also like to get rid of the exponential space which is necessary in the enumeration algorithm by saturation. To tackle this question we need to restrict the set of saturation operations we consider. An element will be a vector over some finite set and in most of this article, we require the saturation operations to act \emph{coefficient-wise} and in the same way on each coefficient. We prove that, when the vector is over the boolean domain, every possible saturation can be computed with polynomial delay. To do that we study a decision version of our problem, denoted by $\Mem_{\cF}$: given a vector $v$ and a set of vectors $\ccS$ decide whether $v$ belongs to the closure of $\ccS$ by the operations of $\cF$. We prove that $\Mem_{\cF} \in \P$ for all set of operations $\cF$ over the boolean domain.

When the domain is boolean, the problem can be reformulated in terms of set systems or hypergraphs. It is equivalent to the generation of the smallest hypergraph which contains a given hypergraph and which is closed under some operations.
We show how to efficiently compute the closure of a hypergraph by any family of set operations (any operation that is the composition of unions, intersections and complementations) on the hyperedges.
Some particular cases were already known previously. For instance, it is well
known that if a family of subsets ordered by inclusion forms a lattice, then the
set of so called \emph{meet irreducible elements} is a generator with respect to the intersection operation (all other sets can be expressed as the intersections of some meet irreducible elements).
In general, knowing how to compute a closure may serve as a good tool to design other enumeration algorithms. One only has to express an enumeration problem as the closure of some sufficiently small and easy to compute set of elements (called a generator) and then to apply the algorithms presented in this article.

The closure computation is also related to constraint satisfaction problems (CSP). Indeed, the set of vectors can be seen as a relation $R$ and the problem of generating its closure by some operations $\cF$ is equivalent to the computation of the smallest relation $R'$ containing $R$ and for which all functions of $\cF$ are polymorphisms of $R'$.
There are several works related to the enumeration in the context of CSP. They deal with the enumeration of solutions of a CSP with polynomial delay~\cite{creignou1997generating,bulatov2012enumerating}.
The simplest such result~\cite{creignou1997generating} states that in the boolean case, there is a polynomial delay algorithm if and only if the constraint language is Horn, anti-Horn, bijunctive or affine.
While those works deal with the enumeration of solutions of CSPs, this paper focuses on finding the closure of relations. However, we use tools from CSPs such as  Post's lattice~\cite{post1941two}, used by Schaefer in his seminal paper~\cite{schaefer1978complexity}, and the Baker-Pixley theorem~\cite{baker1975polynomial}.

The main theorem of this article settles the complexity of a whole family of decision problems and implies, quite surprisingly, that the backtrack search is enough to obtain a polynomial delay algorithm to enumerate any closure of boolean vectors. For all these enumeration problems, compared to the naive saturation algorithm, our method has a better time complexity (even from a practical point of view) and a better space complexity (polynomial rather than exponential). Moreover, besides the generic enumeration algorithm, we give for each closure rule an algorithm with the best possible complexity. In doing so, we illustrate several classical methods used to enumerate objects such as amortized backtrack search, hill climbing, Gray code $\dots$
It is interesting to note that most algorithms we provide have a delay polynomial in the maximum size of a solution we generate.
However some are polynomial in the instance which can be much larger than the size of a solution, for instance the algorithm computing the closure of sets by intersection as explained in Section~\ref{sec:monotone}. In that case, we provide a reduction from the problem of generating the assignments of a monotone $DNF$ formula which suggests that the delay must depend on the instance size.

In a second part of the article, we generalize the set of operators used to compute a closure. The aim is to
capture more interesting problems and to better understand the difference between polynomial incremental time and polynomial delay.
The first generalization is to consider larger domains for the input vectors. In that setting, the problem $\Mem_{\cF}$ is $\NP$-complete for some $\cF$ and we are not able to settle the question in general.
The second generalization is to allow the operators to act differently on each coefficient. We prove that $\Mem_{\cF} \in \P$ when the operators are extremely simple: they change only a single coefficient of the vector and leave the rest unchanged. However, allowing the operators to act on three coefficients is already enough to make $\Mem_{\cF}$ $\NP$~complete.

It is classical in enumeration to try to reduce the number of generated objects, which can be very large, by requiring additional properties. The most classical properties are the maximality or minimality for inclusion, since this is often compatible with other notion of optimality of solutions. Hence algorithms are known for maximal matchings~\cite{uno1997algorithms}, maximal cliques~\cite{JohnsonP88}, minimal transversals~\cite{eiter1995identifying} $\dots$
Therefore, as a third generalization we propose to enumerate only the minimal or maximal elements of the closures. In these settings, the problems are not automatically in polynomial incremental time, since no saturation algorithm generate the maximal or minimal elements only. We prove that either these problems have a polynomial delay algorithm or that they are equivalent to well known problems such as the generation of the circuits of a binary matroid or of the maximal independent sets of a hypergaph which can be solved in polynomial incremental time.

This article is a long version of our previous work see~\cite{mary2016efficient}.
The proofs have been improved and are more detailed, the complexity of several enumeration algorithms
have been improved, sometimes using new techniques and more lower bounds are provided. The subject of the last two sections (non uniform operators and maximal sets) have been proposed in~\cite{wepa2016} and are new for the most part.

\subsection{Organization of the paper}

In Section \ref{sec:preliminary}, we define enumeration complexity, the closure problem $\EnumClo$ and the backtrack search.
In Section \ref{sec:boolean}, we use Post's lattice, restricted through suitable reductions between clones, to
prove that $\Mem_{\cF}$ is polynomial for all set of binary operations $\cF$.
It turns out that there are only a few types of closure operations to study: the monotone operations (Section~\ref{sec:monotone}), the addition over $\mathbb{F}_2$ (Section~\ref{sec:algebra}), the set of all operations (Section~\ref{sec:all}), two infinite hierarchies related to the majority function (Section~\ref{sec:threshold}) and the limit cases of the previous hierarchies (Section~\ref{sec:limit}). For all these closure operations we give for  $\EnumClo$ both a generic backtrack algorithm in polynomial delay and also an ad hoc efficient enumeration algorithm, sometimes with a conditional lower bound matching its complexity.
In Section~\ref{sec:largerdomains}, we give polynomial delay algorithms for three classes of closure operations over any domain and prove that the backtrack search we use in the boolean case fails.
In Section~\ref{sec:nonuniform}, we enrich the set of possible operations by operators acting on a single coefficient, and prove that the problem $\EnumClo$ is still in polynomial delay for any set of such operators.
Finally, in Section~\ref{sec:max} we try to enumerate only the minimal or maximal elements of a closure and we show it is either trivially in polynomial delay or equivalent to famous enumeration problems for which no polynomial delay algorithm are known such as enumerating the circuits of a binary matroid or the maximal independent sets of a $k$ uniform hypergraph.

\section{Preliminary}\label{sec:preliminary}

Given $n\in \mathbb{N}$, $[n]$ denotes the set $\{1,...,n\}$.
For a set $D$ and a vector $v\in D^{n}$, we denote by $v_i$ the $i^{\text{th}}$ coordinate of $v$. Let $i,j\in [n]$, we denote by $v_{i,j}$ the vector $(v_i,v_j)$. More generally,  for a subset $I=\{i_1,...,i_k\}$ of $[n]$ with $i_1<...<i_k$ we denote by $v_I$ the vector $(v_{i_1},..., v_{i_k})$. Let $\ccS$ be a set  of vectors. We denote by $\ccS_I$ the set $\{ v_I \mid v \in \ccS\}$, we say that $v_i$ is a \emph{projection} of $v$ and $\ccS_I$ is a projection of $\ccS$. The characteristic vector $v$ of a subset $E$ of $[n]$, is the vector in $\{0,1\}^{n}$ such that $v_i=1$ if and only if $i\in X$.

\subsection{Complexity}

In this section, we recall basic definitions about enumeration problems and their complexity, for further details and examples see~\cite{phd_strozecki}.

 Let $\Sigma$ be some finite alphabet. An \emph{enumeration problem} is a function $A$ from $\Sigma^*$ to $\mathcal{P}(\Sigma^*)$.
 That is to each input word, $A$ associates a set of words. An algorithm which solves
 the enumeration problem $A$ takes any input word $w$ and produces the set $A(w)$ word by word and \emph{without redundancies}.
 We always require the sets $A(w)$ to be finite. We may also ask $A(w)$ to contain only words of polynomial size in the size of $w$ and  that one can test whether an element belongs to $A(w)$  in polynomial time. If those two conditions hold, the problem is in the class $\EnumP$ which is the counterpart of $\NP$ for enumeration. Because of this relationship to $\NP$, we often call solutions the elements we enumerate.

The computational model is the random access machine model (RAM) with addition, subtraction
and multiplication as its basic arithmetic operations. We have additional output registers,
and when a special OUTPUT instruction is executed, the contents of the output registers is outputted and considered as an element of the outputted set. The RAM model is chosen over Turing machines, because it allows to use data structures more efficiently which impacts the complexity measure used in enumeration.

The \emph{delay} is the time between the productions of two consecutive solutions.
Usually we want to bound the delay of an algorithm for all pairs of consecutive solutions
and for all inputs of the same size. If this delay is polynomial in the size of the input, then
we say that the algorithm is in \emph{polynomial delay} and the problem is in the class $\DelayP$.
If the delay is polynomial in the input and the number of already generated solutions, we say that the algorithm is in
\emph{incremental delay} and the problem is in the class $\IncP$. By definition we have $\DelayP \subset \IncP$.
Moreover, if we assume the exponential time hypothesis, $(\DelayP \cap \EnumP) \neq (\IncP \cap \EnumP)$~\cite{capelli2018incremental}.
In practice problems in $\DelayP$ are considered to be tractable, because the total time is linear in the number of solutions generated and because the solutions can be regularly generated. Moreover, these problems often enjoy algorithms with polynomial space
in the size of the input, which is not the case for most polynomial incremental algorithms. Note that in a polynomial delay algorithm we allow a \emph{polynomial precomputation step}, usually to set up data structures, which is not taken into account in the delay. This is why we can have a delay \emph{smaller than the size of the input}.

For $S$ a set of boolean formulas, we define the enumeration problem $\enumSat{S}$ as the function which associates to a formula of $S$ the set of its models. One classical example is when $S$ is the set of formulas in conjunctive normal form denoted by $CNF$, and the associated problem is $\enumSat{CNF}$ or equivalently $\enumSat{SAT}$.
In this article, we consider other families of boolean formulas: the Horn formulas $HORN$, the bijunctive formulas $2SAT$, the formulas in disjunctive normal form $DNF$ and the monotone formulas in disjunctive normal form $MONDNF$.

To compare two enumeration problems, we use often the notion of \emph{parsimonious reduction} from counting complexity as it is adapted to enumeration complexity contrarily to Cook or Karp reductions. Given two enumeration problems $A$ and $B$, we say that there is a parsimonious reduction from $A$ to $B$ if there are two polynomial time functions $f$ and $g$ such that for all instances $x$ of $A$, $g$ maps bijectively $B(f(x))$ to $A(x)$. The problem $\enumSat{CNF}$ is $\EnumP$ complete for parsimonious reductions by adapting the classic proof of Cook~\cite{cook1971complexity}. In this article, we need a slightly more general notion of reduction,
such that a solution of $B$ may yield several solutions of $A$ as long as they can be efficiently enumerated, as introduced in~\cite{mary2013enumeration}.

\begin{definition}
 Let $A$ and $B$ be two enumeration problems. We say that there is a \emph{polynomial delay reduction} from $A$ to $B$
 if there is a polynomial time function $f$ and $C$ an enumeration problem in $\DelayP$ such that for all instances $x$ of $A$, we have:
 \begin{itemize}
  \item  if $y \neq z$ and $y,z \in B(f(x))$ then $C(xy) \cap C(xz) = \emptyset$
  \item  $A(x) = \displaystyle{\bigcup_{ y \in B(f(x))} C(xy)}$
 \end{itemize}
 where $xy$ is the concatenation of $x$ and $y$.
\end{definition}

When a problem $A$ reduces to $B$ and $B$ reduces to $A$, we say that $A$ and $B$ are \emph{equivalent}.
Remark that the class $\DelayP$ is closed under polynomial delay reduction. We will use tight reduction, where
the size of $f(x)$ is closely related to $x$ and $C$ has a linear delay to obtain lower and upper bounds on the complexity
of several enumeration problems.

We now explain a very classical and natural enumeration method called the \emph{Backtrack Search} (sometimes also called the \emph{flashlight method}) used in many previous articles~\cite{read1975bounds,strozecki2013enumerating}. It can be used to solve all auto-reducible problems, in particular $\enumSat{C}$ when $C$ is a set of formulas solvable in polynomial time~\cite{creignou1997generating} and closed by partial assignment of variables. We represent the solutions we want to enumerate as vectors of size $n$ and coefficients in $D$. In practice solutions are often subsets of $[n]$ which means that $D = \{0,1\}$ and the vector is the characteristic vector of the subset.

The  enumeration algorithm is a depth first traversal of a tree whose nodes are partial solutions.
The nodes of the tree are all vectors $v$ of size $l$, for all $l \leq n$, such that
$v = w_{[l]}$ and $w$ is a solution. The children of the node $v$ are
the vectors of size $l+1$, which restricted to $[l]$ are equal to $v$.
The leaves of this tree are the solutions of our problem, therefore a depth first traversal visits all
leaves and thus outputs all solutions. Since a branch of the tree is of size $n$, it is enough to find the children of a node in a time polynomial in $n$ to obtain a polynomial delay. The delay also depends linearly on $|D|$, but in the rest of the paper $|D|$ will be constant.
Therefore the backtrack search is in polynomial delay if and only if the the following \emph{decision} problem is in $\P$:
given $v$ of size $l$ is there $w$ a solution such that $v = w_{[l]}$? This problem is called the \emph{extension problem} associated to the enumeration problem.

\begin{proposition}\label{prop:partialsolutions}
 Given an enumeration problem $A$, such that for all $w$, $A(w)$ is a set of vectors of size $n$ and coefficients in $D$,
 with $n$ and $|D|$ polynomially related to $|w|$. If the extension problem associated to $A$ is in $\P$, then $A$ is in $\DelayP$.
\end{proposition}

We will see in the next part, that the complexity of solving the extension problem can be amortized
over a whole branch of the tree, since we solve it many times, using well chosen data structures.

There is a second classical enumeration method to design a polynomial delay algorithm, named the \emph{supergraph} method. The idea is to organize the solutions (and not the partial solutions) as a strongly connected digraph instead of a tree, and to traverse this supergraph. For that we should be able to visit all the successors of a node in polynomial time. To avoid the storage of the nodes of the supergraph, the \emph{reverse search} method is often used~\cite{avis1996reverse}. It consists on defining a canonical parent computable in polynomial time for each node and thus defining a spanning arborescence of the supergraph.  This method may have a better delay than the backtrack search, because the traversal goes over solutions only.

\subsection{Closure of families by set operations}
%
We fix a finite domain $D$. Given a $t$-ary operation $f$  (a function from $D^t$ to $D$), $f$ can be naturally extended to a $t$-ary operation over vectors of the same size. Let $(v^{1},\dots v^{t})$ be a $t$-uple of vectors of size $n$, $f$ acts coefficient-wise on it, that is for all $i\leq n$, $f(v^{1},\dots, v^{t})_i = f(v^{1}_i,\dots, v^{t}_i)$.

\begin{definition}
Let $\cF$ be a finite set of operations over $D$.
Let $\ccS$ be a set of vectors over $D$.
Let $\cF^i(\ccS) = \{f(v_1,\dots, v_t) \mid v_1,\dots, v_t \in \cF^{i-1}(S) \text{ and } f\in \cF \} \cup \cF^{i-1}(\ccS)$ and $\cF^0(\ccS) = \ccS$.
The closure of $\ccS$ by $\cF$ is $\Cl_\cF(\ccS) = \displaystyle{\cup_i \cF^i(\ccS)}$.
\end{definition}

Notice that $\Cl_{\cF}(\ccS)$ is also the smallest set which contains $\ccS$ and which is closed by the operations of $\cF$.
The set $\Cl_\cF(\ccS)$ is invariant under the operations of $\cF$: these operations are called \emph{polymorphisms} of the set $\Cl_\cF(\ccS)$, a notion which comes from universal algebra.

As an illustration, assume that $D = \{0,1\}$ and that  $\cF = \{ \vee \}$.
Then the elements of $\ccS$ can be seen as subsets of $[n]$ (each vector of size $n$ is the characteristic vector of a subset of $[n]$) and $\Cl_{\{\vee\}}(\ccS)$ is the closure by union of all
sets in $S$. Let $\ccS = \{ \{1,2,4\},  \{2,3\},  \{1,3\} \}$ then
$$\Cl_{\{\vee\}}(\ccS)= \{ \{1,2,4\},  \{1,2,3,4\}, \{2,3\},  \{1,3\},  \{1,2,3\}\}.$$ Note that $\Cl_{\{\vee\}}(\ccS)$ is indeed closed by union, that is $\vee$ is a polymorphism of  $\Cl_{\{\vee\}}(\ccS)$.

The problem we try to solve in this article, for all set of operations $\cF$ over $D$,
is $\EnumClo_\cF$: given a set of vectors $\ccS$ list the elements of $\Cl_\cF(\ccS)$.
We always denote the size of the vectors of $\ccS$ by $n$ and the cardinality of $\ccS$
by $m$. A naive saturation algorithm solves $\EnumClo_\cF$: the elements of $\cF$ are applied to
all tuples of $\ccS$ and any new element is added to $\ccS$, the algorithm stops when no new element can be produced.
If the largest arity of an operation in $\cF$ is $t$ then this algorithm adds a new element to a set of $k$ elements (or stops)
in time $O(k^tn)$, therefore it is in $\IncP$.

The aim of this article is to find better algorithms to solve $\EnumClo_\cF$, more specifically polynomial
delay algorithms with polynomial memory. To do that, let us introduce the extension problem associated to a set of operations $\cF$, denoted by $\ExtClo_\cF$: given $\ccS$ a set of vectors of size $n$,  and a vector $v$ of size $l \leq n$, is there a vector $v' \in \Cl_\cF(\ccS)$ such that $v'_{[l]} = v$. Solving this problem in polynomial time yields a polynomial delay algorithm using backtrack search by Proposition~\ref{prop:partialsolutions}. In fact, we need only to solve the simpler membership problem, denoted by $\Mem_\cF$:
given $\ccS$ a set of vectors of size $n$,  and a vector $v$ of size $n$ does $v$ belongs to $\Cl_\cF(\ccS)$?

\begin{proposition}\label{prop:decisionenum}
 If $\Mem_\cF \in \P$ then $\EnumClo_\cF \in \DelayP$.
\end{proposition}

\begin{proof}
$\ExtClo_\cF$ can be reduced to $\Mem_\cF$.
Indeed, given a vector $v$ of size $l$, because the operations of $\cF$ act coordinate-wise,
the two following predicates are equivalent:
\begin{itemize}
 \item  $\exists v' \in \Cl_\cF(\ccS)$ such that $v'_{[l]} = v$
 \item $v \in  \Cl_\cF(\ccS_{[l]})$
\end{itemize}
Therefore if $\Mem_\cF \in \P$ then we have also  $\ExtClo_\cF \in \P$.
 We use Proposition~\ref{prop:partialsolutions} to conclude.
\end{proof}

We have introduced an infinite family of problems, whose complexity we want to determine.
Several families of operations may always produce the same closure. To deal with that, we need to introduce the notion of functional clone. We write clone instead of functional clone in the rest of the paper since there is no ambiguity in our context.

\begin{definition}
 Let $\cF$ be a finite set of operations over $D$, the functional clone generated by
 $\cF$, denoted by $<\cF>$, is the set of operations obtained by any composition of the operations of $\cF$ and of the projections $\pi_k^n : D^n \to D$
 defined by $\pi_k^n(x_1,\dots,x_n) = x_k$.
\end{definition}

This notion is useful, because two sets of functions which generate the same clone applied to the same set produce the same closure.

\begin{lemma}
For all set of operations $\cF$ and all set of vectors $\ccS$, $\Cl_\cF(\ccS) = \Cl_{<\cF>}(\ccS)$.
\end{lemma}

The number of clones over $D$ is infinite even when $D$ is the boolean domain (of size $2$).
However, in this case the clones form a countable lattice, called  Post's lattice~\cite{post1941two}.
Moreover there is a \emph{finite} number of well described clones plus a few number of infinite families of very regular clones.

\section{The Boolean Domain}\label{sec:boolean}

In this part, we prove the main theorem on the complexity of $\Mem_\cF$, when the domain is boolean.
An instance of one such problem, denoted  by $\ccS$, will be indifferently seen as a set of vectors of size $n$ or a set of subsets of $[n]$.

\begin{theorem}\label{th:main}
 Let $\cF$ be any fixed finite set of operations over the boolean domain, then $\Mem_\cF \in \P$ and $\EnumClo_\cF \in \DelayP$.
\end{theorem}

There is also a uniform version of $\Mem_\cF$, where $\cF$ is given as input.
It turns out that this problem is $\NP$-hard as proven in Section \ref{sec:threshold}.

To prove Theorem~\ref{th:main}, we prove that $\Mem_\cF \in \P$,
for each clone $\cF$ of the Post's lattice. Since the Post's lattice contains many classes, we first show that
many of them are equivalent with regard to the complexity of $\Mem_\cF$: for some $\cF$
the problem $\Mem_\cF$ can be reduced to $\Mem_\mathcal{G}$ where $\mathcal{G}$ is another clone obtained from a simple transformation of $\cF$. This reduces the number of cases we need to consider.

To an operation $f$ we can associate its dual $\overline{f}$ defined by $\overline{f}(s_1,\dots,s_t) = \neg{f(\neg{s_1},\dots,\neg{s_t})}$. If $\cF$ is a set of operations, $\overline{\cF}$ is the set of duals of operation in $\cF$.
We denote by $\zero$ and $\one$ the constant functions which always return $0$ and $1$. By a slight abuse of notation,
we also denote by $\zero$ the all zero vector and by $\one$ the all one vector.

\begin{proposition}\label{prop:postsimpl}
The following problems can be reduced to $\Mem_{\cF}$ by a polynomial time parsimonious reduction:
\begin{enumerate}
 \item $\Mem_{\cF \cup \{\zero\}}$, $\Mem_{\cF \cup \{\one\}}$, $\Mem_{\cF \cup \{\zero,\one\}}$
 \item $\Mem_{\overline{\cF}} $
 \item $\Mem_{\cF\cup \{\neg \}} $ when $\cF = \overline{\cF}$
\end{enumerate}
\end{proposition}

\begin{proof}
The reductions follow easily from these observations:
\begin{enumerate}
 \item $\Cl_{\cF\cup \{f\}}(\ccS) = \Cl_{\cF}(\ccS \cup \{f\})$ for $f=\zero$ or $f=\one$ and  $\ccS \neq \emptyset$.
 \item $\Cl_{\overline{\cF}}(\ccS) =  \overline{\Cl_{\cF}(\overline{\ccS})} $ where $\overline{\ccS}$ denotes the set of negation of vectors in $\ccS$.
 \item $\Cl_{\cF\cup \{\neg \}}(\ccS) = \Cl_{\cF}(\ccS \cup \overline{\ccS}) $  since for every $f \in \cF$, there exists $g\in \cF $ such that $\neg f(v_1,\dots,v_t) = \overline{f}(\neg{v_1},\dots,\neg{v_t})=g(\neg{v_1},\dots,\neg{v_t})$.
\end{enumerate}
\end{proof}

In Figure~\ref{fig:post}, we represent the clones which cannot be reduced to another one using Proposition~\ref{prop:postsimpl} and their bases. We also represent $BF$ the clone of all function since it is useful in several reductions presented in the rest of the section. For a modern presentation of all boolean clones, their bases and the Post's lattice see~\cite{bohler2002boolean,reith2003optimal}. Each clone correspond to a case to settle in our proof of Theorem~\ref{th:main}. We can further group clones by the algorithm used to solve their associated enumeration problems
and we represent these by ellipses enclosing similar clones in Figure~\ref{fig:post}.
Each ellipse in Figure~\ref{fig:post} corresponds to a subsection of this section.

\begin{figure}[h]
\begin{center}
\begin{tabular}{c c}

\begin{tabular}{|l|l|}\hline
Clone & Base\\\hline
$I_2$ & $\emptyset$\\\hline
$L_2$ & $x+y+z$\\\hline
$L_0$ & $x+y$\\\hline
$E_2$ & $\wedge$\\\hline
$S_{10}$ & $x \wedge( y \vee z) $ \\\hline
$S^2_{10}$ &$maj, x \wedge( y \vee z) $\\\hline
$S^k_{10}$ &$Th_k^{k+1}, \,k\geq 3$\\\hline
$S_{12}$ & $x \wedge (y \to z)$\\\hline
$S^k_{12}$ & $Th_k^{k+1}, x \wedge (y \to z)$\\\hline
$D_2$ & $maj$\\\hline
$D_1$ & $maj, x+y+z$\\\hline
$M_2$ & $\vee,\wedge$ \\\hline
$R$ & $x\,?\,y\,:\,z$\\\hline
$R_0$ & $\vee, +$\\\hline
$BF$ & $\vee, \neg$\\\hline
 \end{tabular}& \hspace{1cm}
 \parbox{5cm}{
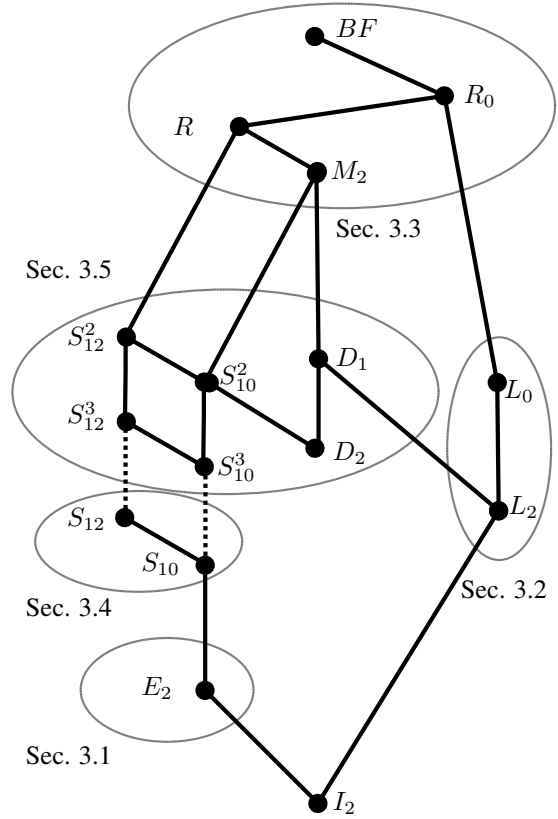}
\end{tabular}
 \caption{The reduced Post's lattice, upward edges represent inclusions of clones \label{fig:post}}
\end{center}
\end{figure}

To further simplify the enumeration problems we solve, we remark the following.
When there is an index $i$ such that the value of $v_i$ for all elements $v$ to be enumerated is determined by the value of another index or is not constrained, then we can project it out. The formal statement of this fact is given in the next proposition, and the proof is by obvious polynomial delay reductions. Note that these polynomial delay reductions decrease the size of the instance and incurs an overhead for each solution which is constant. Since these reductions can be applied at most once per index, the general overhead for each solution is bounded by the size of the solution.

\begin{proposition}\label{prop:red}
The problem $\EnumClo_{\cF}$ can be reduced by a polynomial delay reduction to
$\EnumClo_{\cF}$ where the instances $\ccS$ do not satisfy these two properties:
 \begin{itemize}
  \item  $\exists i,j \in [n], \, \exists f :\{0,1\}\rightarrow \{0,1\}, \,\forall v\in \Cl_{\cF}(\ccS), \, v_i = f(v_j)$
  \item $\exists i \in [n], \, \forall v\in \Cl_{\cF}(\ccS), \, \exists v'\in \Cl_{\cF}(\ccS), \, v_{[n]\setminus i} = v'_{[n]\setminus i}$ and $v_i = \neg v'_i$
 \end{itemize}
\end{proposition}
\begin{proof}
 In the two cases, the reduction is to $\EnumClo_{\cF}$ on the instance $\ccS_{[n]\setminus i}$.
 Then each solution $v$ is extended by $f(v_j)$ in the first case or by $0$ and $1$ in the second case.
\end{proof}

\subsection{Conjunction}\label{sec:monotone}

We first study one of the simplest clones: $E_2 = <\wedge>$.
We give an elementary proof that $\Mem_{E_2} \in \P$, then we explain how to
obtain a good delay for $\EnumClo_{E_2}$.
For a binary vector $v$, let us denote by $\0 (v)$ (resp.  $\1 (v)$) the  set of indices $i$ for which $v_i=0$ (resp. $v_i=1$).

\begin{proposition}
$\Mem_{E_2} \in \P$.
\end{proposition}
\begin{proof}
 Let $\ccS$ be a set of boolean vectors, if we apply $\wedge$ to a couple of vectors in $\ccS$ it produces the intersection of two vectors when seen as sets.
 Since the intersection operation is associative and commutative, $\Cl_{E_2}(\ccS)$ is the set of arbitrary intersections of elements of  $\ccS$.
 Let $v$ be a vector and let $\ccS_1$ be the set  $\{w \in \ccS \mid w_{\1 (v)} = \one \} $.
 Assume now that $v$ can be obtained as an intersection of elements $v_1,\dots,v_t$, those elements must be in $\ccS_1$ because of the monotonicity of the intersection for the inclusion.
 On the other hand, by definition of $\ccS_1$, $v$ is contained in  $\displaystyle{\bigcap_{w \in \ccS_1} w}$.
  Therefore, $v \in \Cl_{E_2}(\ccS)$ if and only if $v = \displaystyle{\bigcap_{w \in \ccS_1} w}$.
 This intersection can be computed in time $O(mn)$ which concludes the proof.
\end{proof}

By Proposition~\ref{prop:partialsolutions}, we can turn the algorithm for $\Mem_{E_2}$ into an enumeration algorithm
for $\EnumClo_{E_2}$ with delay $O(mn^2)$. We explain in the next proposition how to reduce this delay to $O(mn)$.

\begin{proposition}\label{prop:unionfast}
There is an algorithm solving $\EnumClo_{E_2}$ with delay $O(mn)$.
\end{proposition}
\begin{proof}
 We use the backtrack search described in Proposition \ref{prop:partialsolutions} but we maintain data structures which allow to decide $\Mem_{E_2}$ quickly.
 Let $\ccS$ be the input set of $m$ vectors of size $n$.
 During the traversal of the tree we update the partial solution $p$, represented by an array of size $n$ which stores
 whether $p_i = 1$, $p_i=0$  or is yet undefined.

 A vector $v$ of $\ccS$ is compatible with the partial solution if $\1(p) \subseteq \1(v)$.
 We maintain an array $COMP$ indexed by the sets of $\ccS$, which stores whether each vector of $\ccS$ is compatible or not with the current partial solution.
 Finally we update an array $COUNT$, such that $COUNT[i]$ is  the number of compatible vectors $v \in \ccS$ such that $v_i =0$.
 Notice that a partial solution $p$ can be extended into a vector of $\Cl_{E_2}(\ccS)$ if and only if for all $i \in \0(p)$ $COUNT[i]>0$, the solution is then the intersection of all compatible vectors.

 At each step of the traversal, we select an index $i$ such that $p_i$ is undefined and we set first $p_i= 0$ then  $p_i=1$.
 When we set $p_i=0$, there is no change to do in $COUNT$ and $COMP$ and we can check whether this extended partial solution is correct by checking if $COUNT[i]>0$ in constant time.
 When we set $p_i=1$, we need to update $COMP$ by removing from it every vector $v$ such that $v_i =0$.
 Each time we remove such a vector $v$, we decrement $COUNT[j]$ for all $j$ such that $v_j = 0$.
 If there is a $j$ such that  $COUNT[j]$ is decremented to $0$ then the extension of $p$ by $p_i = 1$ is not possible.

 When we traverse a whole branch of the tree of partial solutions during the backtrack search,
 we set $p_i=1$ for each $i$ at most once and then we need to remove each vector from $COMP$ at most once.
 Therefore the total number of operations we do to maintain $COMP$ and $COUNT$ is $O(mn)$ and so is the delay.
\end{proof}

We now relate the problem $\EnumClo_{E_2}$ to $\enumSat{DNF}$ which is the problem of listing all satisfying assignments of a formula in disjunctive normal form. The problem $\enumSat{DNF}$ appears in several contexts such as the enumeration of satisfying assignments of an existential formula with second order variables~\cite{DurandS11} or knowledge compilation~\cite{capelli2016structural}.
The best algorithms to solve $\enumSat{DNF}$ have a delay of $O(mn)$ where $n$ is the number of variables and $m$ the number of clauses. The following question is open and seems hard to solve:\\
Is there an algorithm for $\enumSat{DNF}$ with delay polynomial in $n$ only or which has a sublinear dependency on $m$ ?
In this paper we use $\enumSat{MONDNF}$, the variant with only positive variables, as a \emph{hard} problem, that is we conjecture  it cannot be solved with a delay better than $O(m)$. Then, in the spirit of fine grained complexity~\cite{williams2010subcubic}, it can be used to prove conditional lower bounds on the delay of other enumeration problems.

Note that if we ask for all subsets of the sets in $\ccS$ instead of all intersections, we exactly get the problem of enumerating the solutions of a monotone DNF formula. Moreover, the algorithm of Proposition~\ref{prop:unionfast} is similar to one used to generate the solutions of a DNF formula or of a monotone CNF formula~\cite{murakami2014efficient}.  We now show that we can reduce the problem
$\enumSat{MONDNF}$ to $\EnumClo_{E_2}$. Moreover, an instance of $\enumSat{MONDNF}$ with $n$ variables and $m$ clauses is transformed into an instance of $\EnumClo_{E_2}$ with $mn$ vectors of size $n$. Therefore any improvement in the delay for $\EnumClo_{E_2}$ yields a better delay for $\enumSat{MONDNF}$.

 \begin{proposition}\label{prop:hardDNF}
 There is a polynomial delay reduction from $\enumSat{MONDNF}$ to $\EnumClo_{E_2}$.
 \end{proposition}

 \begin{proof}
 First recall that by the second point of Proposition~\ref{prop:postsimpl}, $\EnumClo_{\{\vee\}}$ can be reduced to $\EnumClo_{E_2}$ with an instance of the exact same size. Let $\phi \equiv \bigvee_{i=1}^{m} C_i$ be a monotone DNF formula over the variables $X = \{x_0,\dots,x_{n-1}\}$. The $C_i$ are clauses over positive variables and can be seen as the set of their variables.
 To $C_i$ we associate the vectors $v^{i,j}$ which are equal to the characteristic vector of $C_i$ except on the coordinate $j$ where it is equal to $1$. Let $\ccS = \{ v^{i,j} \mid i \in [m],j \in [n]\}$. Notice that a solution of $\phi$ is any assignment which contains all variables of some clause $C_i$. The set $\Cl_{\{\vee\}}(\ccS)$ is also the set of all sets containing some clause which proves the reduction from $\enumSat{MONDNF}$ to $\EnumClo_{\{\vee\}}$ and thus to $\EnumClo_{D_2}$.
\end{proof}

\subsection{Algebraic operations}\label{sec:algebra}
We first deal with the clone $L_0 = <+>$ where $+$ is the boolean addition (or equivalently the boolean operation $XOR$).
Note that $\Cl_{L_0}(\ccS)$ is the vector space generated by the vectors in $\ccS$.
Seen as an operation on sets, $+$ is the symmetric difference of the two sets.

\begin{proposition}
 $\Mem_{L_0} \in \P$.
\end{proposition}
\begin{proof}
Let $\ccS$ be the set of input vectors, let $v$ be a vector and let $A$ be the matrix whose rows are the elements of $\ccS$.
The vector $v$ is in $\Cl_{L_0}(\ccS)$ if and only if there is a solution to $Ax =v$.
Solving a linear system over $\mathbb{F}_2$ can be done in polynomial time which proves the proposition.
\end{proof}

The previous proposition yields a polynomial delay algorithm by applying Proposition~\ref{prop:partialsolutions}.
One can get a better delay, by computing in polynomial time a maximal free family $M$ of $\ccS$, which is
a basis of $\Cl_{L_0}(\ccS)$. The basis $M$ is a succinct representation of $\Cl_{L_0}(\ccS)$. One can generate all elements of $\Cl_{L_0}(\ccS)$
by going over all possible subsets of elements of $M$ and summing them. The subsets can be enumerated in constant time by using Gray code enumeration (see~\cite{knuth2011combinatorial}).
The sum can be done in time $n$ by adding a single vector since two consecutive sets differ by a single element in the Gray code order. Therefore we have, after the polynomial time computation of $M$, an enumeration with delay $O(n)$.
If one allows to output the elements represented in the basis $M$, the algorithm even has constant delay.

The closure by the clone $L_2$ (generated by the sum modulo two of three elements) corresponds to an affine space rather than a vector space. Hence, we can easily extend the previous result to  $L_2$.
\begin{proposition}
 $\Mem_{L_2} \in \P$.
\end{proposition}
\begin{proof}
 Since the sum of three elements is associative and commutative, the vectors in $\Cl_{L_2}(\ccS)$ are the sum of an odd number of vectors in $\ccS$. In other words $v \in \Cl_{L_2}(\ccS)$ if and only if there is a $x$ such that $Ax = v$ and the Hamming weight of $x$ is odd. One can compute a basis $B$ of the vector space of the solutions to the equation $Ax = v$. If all elements of $B$ have even Hamming weight, then their sums also have even Hamming weight. Therefore
 $v \in \Cl_{L_2}(\ccS)$ if and only if there is an element in $B$ with odd Hamming weight, which can be decided in polynomial time.
\end{proof}

\subsection{Boolean algebras}\label{sec:all}

In this subsection, we deal with the largest clones of our reduced Post lattice: $M_2 = <\wedge,\vee>$, $BF = <\vee,\neg>$, $R_0= <\vee,+>$ and $R = <x \,?\, y \,:\, z>$, where $x \,?\, y \,:\, z$ is the "if then else" operator which is equal to $(\neg x \vee y) \wedge (x \vee z) $.

\begin{proposition}
$\Mem_{M_2} \in \P$.
\end{proposition}
\begin{proof}
Let $\ccS$ be a vector set and for all $i\in [n]$, let $x^i = \displaystyle{\bigwedge_{v\in \ccS \text{ s.t. } v_i=1} v}$, we call
$x^i$ an atom. We show that a vector $u$ belongs to $\Cl_{M_2}(\ccS)$ if and only if $u = \bigvee\limits_{i\in \1(u)} x^i$. By definition of  $\Cl_{M_2}(\ccS)$, it contains the  atoms since they are intersections of elements in $\ccS$ and it contains their unions $u = \bigvee\limits_{i\in \1(u)} x^i$. Since the intersection distributes with the union, we can write any element $v$ in $\Cl_{M_2}(\ccS)$ as $v = \bigvee_i \left(\bigwedge_j s^{i,j} \right)$ where $s^{i,j} \in \ccS$.
Hence it is enough to show that $v^ i = \bigwedge_j s^{i,j} $ is an union of atoms. Let $k$ be an index such that $v^i_k = 1$,
then for all vectors $s^{i,j}$, $s^{i,j}_k = 1$. It implies that seen as a set, $x^k$ is included in $v^i$ and that $v^i =  \bigvee\limits_{k\in \1(v^i)} x^k$ which proves the characterization of $\Cl_{M_2}(\ccS)$.

We can compute the atoms in time $O(mn^2)$, and then to decide whether $v \in \Cl_{M_2}(\ccS)$, one must check
 whether $v$ contains all atoms $x^i$ such that $v_i = 1$ in time $0(n^2)$ which proves the proposition.
\end{proof}

 Applying Proposition~\ref{prop:partialsolutions}, we get an enumeration algorithm with delay $O(mn^3)$.
 Moreover, $\Cl_{M_2}(\ccS) = \Cl_{<\vee>}(\{x^i\}_{i \in [n]})$. Recall that $\EnumClo_{<\vee>}$ can be reduced to $\EnumClo_{E_2}$ by Proposition~\ref{prop:postsimpl} and that $\EnumClo_{E_2}$ has an algorithm with delay $O(mn)$ thanks to Proposition~\ref{prop:unionfast}. Hence the $n$ atoms can be precomputed and their union generated with delay $O(n^2)$ since here $m = n$. We can do better by using the inclusion structure of the $x^i$'s to obtain a $O(n)$ delay.

\begin{proposition}\label{prop:M2}
 $\EnumClo_{M_2}$ can be solved with delay $O(n)$.
\end{proposition}
\begin{proof}
Let $\ccS$ be the input. We first build for all $i$, $x^i = \bigwedge\limits_{v \in \ccS, v_i = 1} v$.
Inclusion is a partial order between the elements $x^i$ seen as sets. We consider a linear extension $T$ of that partial order, i.e. a total order obtained by ordering all the incomparable pairs.
 We then generate all elements of $\Cl_{M_2}(S)$ by a Hill climbing algorithm: we go from one solution
 to another by adding a single $x^i$. Let $v$ be the current solution, we maintain a list $L$ ordered according to $T$ of the indices $i$ of $v$ such that $v_i = 0$. At each step we select $i$ the smallest element of $L$ and we set $v_j =1$ and remove $j$ from $L$ for all
 $j \in x^i$. This produces a new solution in time $O(n)$. We then recursively call the algorithm on this new solution and list.
 When the recursive call is finished, we call the algorithm on $v$ and $L \setminus \{i\}$.

 This algorithm is correct, because the solutions generated in the two recursive calls are disjoint.
 Indeed, in the second call $v_i$ will always be $0$, because all indices in $L$ are bigger than $i$ in $T$.
 It means that $x^j$ for $j \in L$ is either smaller or incomparable. Since $x^i$ is the smallest element with $x^i_i =1$ it implies that $x^j_i = 0$.
\end{proof}

The problem $\EnumClo_{BF}$, which can be reduced to $\EnumClo_{M_2}$ by Proposition~\ref{prop:allreduc} is in fact even easier to enumerate. Let $\ccS$ be a set of vectors, let $X^i = \{ v \mid v \in \ccS, v_i = 1 \} \cup \{ \neg v \mid v \in \ccS, v_i = 0 \}$ and
let $x^i =  \bigwedge_{v\in X^i} v$. The set $\Cl_{BF}(\ccS)$ is in fact a boolean algebra, whose atoms are the disjoint elements $x^i$.
Indeed, either $x^i_{i,j}= x^j_{i,j}$ and they are equal or $\1_{x^i} \cap \1_{x^j} = \emptyset$.
Let $A = \{ x^i \mid i\in [n]\}$, two distinct unions of elements in $A$ produce distinct elements.
Hence by enumerating all possible subsets of $A$ with a Gray code, we can generate $\Cl_{BF}(\ccS)$ with delay $O(n)$ or $O(1)$ if groups of always equal coefficients are represented by a single coefficient.

We now prove that the closures by the clones $R$ and $R_0$ are equal to the closure by $BF$ up to
some coefficients which are fixed to $0$ or $1$, thus they are as easy to enumerate as stated in the next proposition.

\begin{proposition}\label{prop:allreduc}
 The problems $\EnumClo_{R}$ and $\EnumClo_{R_0}$ can be reduced to $\EnumClo_{BF}$ by a polynomial delay reduction.
\end{proposition}
\begin{proof}
First notice that both $R$ and $R_0$ contains $M_2$ which means that they contain $\vee$ and $\wedge$.
 By Proposition ~\ref{prop:red} we can assume that for all $i$, there is $u \in \ccS$ such that $u_i = 1$ and $v$
 such that $v_i =0$. Since $\vee$ is in both clones the vector $\one$ is in their closure. Since the clones also contain $\wedge$, the vector $\zero$ is also in their closures.
 Now notice that $ x \,?\, \zero \,:\, \one = \overline x$ and that $x + \one = \overline{x}$. As a consequence we have $\Cl_{R}(\ccS) = \Cl_{R_0}(\ccS) = \Cl_{BF}(\ccS)$ which proves the proposition.
\end{proof}

\subsection{Limits of the infinite parts}\label{sec:limit}

We deal with the two infinite hierarchies of clones in Subsection~\ref{sec:threshold}.
However, the method for the hierarchy does not directly apply to the clones which are the limits of these hierarchies:
$S_{10} = < x \wedge (y \vee z)>$ and $S_{12} = < x \wedge (y \to z)>$. We show here how we can see them as a union of
closures by $M_2$ or $BF$, which complexity wise makes them similar to $\EnumClo_{E_2}$. As a consequence, they have a polynomial delay algorithm but the delay depends on $m$ the size of $\ccS$ and a reduction from $\enumSat{MONDNF}$ suggests it is hard to get rid of this dependency.

\begin{lemma}\label{lemma:struc}
 Let $\ccS$ be a set of vectors then:
 \begin{itemize}
  \item $\Cl_{S_{10}}(\ccS) = \displaystyle{\bigcup_{s\in \ccS} \Cl_{M_2}(\ccS_{\1(s)})}$
  \item there is a set of indices $I$ such that if $v \in \Cl_{S_{12}}(\ccS_I)$, $v_I = \one$ and
  $$\Cl_{S_{12}}(\ccS_{\bar{I}}) = \displaystyle{\bigcup_{s\in \ccS} \Cl_{BF}(\ccS_{{\bar{I}} \cap \1(s)})}$$

 \end{itemize}
where the elements of $\Cl_{M_2}(\ccS_{\1(s)})$ and of $\Cl_{BF}(\ccS_{\1(s)})$ are extended by zeros outside of $\1(s)$.
\end{lemma}
\begin{proof}
 By induction on the elements of  $\Cl_{S_{10}}(\ccS)$ (respectively $\Cl_{S_{12}}(\ccS)$),
 one proves that their support are always contained in the support of some element of  $\ccS$.

 Hence, for $v\in \Cl_{S_{10}}(\ccS)$,  we assume that $v \subseteq s \in \ccS$.
 Notice that $x \wedge (y \vee y)$ is equal to $x \wedge y$, therefore we have the operation $\wedge$ in $S_{10}$.  If $v_{\1(s)} \in \Cl_{S_{10}}(\ccS_{\1(s)})$, it implies that there is a $v' \in \Cl_{S_{10}}(\ccS)$ such that $v'_{\1(s)} = v_{\1(s)}$. Therefore  $v'\wedge s = v$ and $v \in  \Cl_{S_{10}}(\ccS)$. We have proved $v \in \Cl_{S_{10}}(\ccS)$ if and only if $v_{\1(s)} \in \Cl_{S_{10}}(\ccS_{\1(s)})$ since the other direction is just the projection on $\1(s)$.

  As a consequence, if we consider $\ccS_{\1(s)}$, it contains the vector $\one$. We can simulate $\vee$, since
  $ \one \wedge (x \vee y) = x \vee y$. Therefore $\Cl_{M_2}(\ccS_{\1(s)}) \subseteq \Cl_{S_{10}}(\ccS_{\1(s)})$
  and the reverse inclusion is clear since $M_2 = <\vee,\wedge>$ contains $S_{10}$ which proves the first item of the lemma.

  Assume that there is a coefficient $i$ such that for all $s\in \ccS$, $s_i=1$, then for all $v \in \Cl_{S_{12}}(\ccS)$, $v_i = 1$. We denote by $I$ the set of such coefficients and by the previous remark,
  for all $v$ in $Cl_{S_{12}}(\ccS)$, we have $v_I = \one$.
  By definition, for each $i \in \bar{I}$, there is a $s\in \ccS$ with $s_i=0$. Since $x \wedge (x \to y) = x \wedge y$, the vector $\zero$ is in $\Cl_{S_{12}}(\ccS_{\bar{I}})$.

  As in the first part of the proof, $v \in \Cl_{S_{12}}(\ccS_{\bar{I}})$ if and only if there is $v \subseteq s\in \ccS_{\bar{I}}$, such that $v_{\1(s)} \in \Cl_{S_{12}}(\ccS_{\bar{I} \cap \1(s)})$.
  By definition,  $s_{\1(s)} = \one$. Thus we can simulate the negation in $\Cl_{S_{12}}(\ccS_{\bar{I} \cap\1(s)})$ by $\one \wedge (x \to \zero) = \neg x$. We also simulate the disjunction by  $\one \wedge (\neg x \to y) = x \vee y$. Therefore $\Cl_{S_{12}}(\ccS_{\bar{I} \cap \1(s)})$  contains $\Cl_{BF}(\ccS_{\bar{I} \cap\1(s)})$.
  The reverse inclusion follows from the fact that $BF$ contains $S_{12}$ which proves the second item of the lemma.
\end{proof}

As a consequence of this structural characterization, is it easy to decide the problems $\Mem_{S_{10}}$ and
$\Mem_{S_{12}}$ in time $O(mn)$. Therefore using a backtrack search, we obtain enumeration algorithms with delay $O(mn^2)$.
However, using the structures of $\Cl_{S_{10}}(\ccS)$ and $\Cl_{S_{12}}(\ccS)$ described in Lemma~\ref{lemma:struc} we obtain better algorithms described in the following proposition.

 \begin{proposition}\label{prop:S10S12}
    The problems $\EnumClo_{S_{10}}$ and $\EnumClo_{S_{12}}$ can be solved with delay $O(mn)$.
 \end{proposition}
\begin{proof}
We first prove the result for $\EnumClo_{S_{12}}$.
By Lemma~\ref{lemma:struc}, there is a set $I$, such that $\Cl_{S_{12}}(\ccS_{\bar{I}}) = \cup_{s\in \ccS} \Cl_{BF}(\ccS_{\bar{I} \cap \1(s)}) $ where the elements of $\ccS_{\1(s)}$ are extended by zeros. The coefficients in $I$ are always one and do not matter for an enumeration algorithm.
Moreover, $\EnumClo_{BF}$ can be solved with delay $O(n)$.  In particular, we can output the solutions in lexicographic order, since they are union of disjoint atoms.
If the elements of $\Cl_{BF}(\ccS_{\1(s)}) $ are generated for each $s$ we obtain all the elements we want to output but with \emph{redundancies}. We can obtain an algorithm without redundancies by simulating an enumeration of the elements of $\Cl_{BF}(\ccS_{\1(s)}) $ for all $s$ and always outputting the smallest one. This method is described in~\cite{phd_strozecki} (Lemma $3$, Chapter $2$) and the delay is the sum of delays of the enumeration procedure for each $s$, that is $O(mn)$.

 The proof for $\EnumClo_{S_{10}}$ is similar and follows from Proposition~\ref{prop:M2}
 which shows that the elements of $\Cl_{M_2}(\ccS)$ can be enumerated with delay $O(n)$ in lexicographic order.
\end{proof}

\begin{proposition}
 There is a polynomial delay reduction from $\enumSat{MONDNF}$ to $\EnumClo_{S_{10}}$ and $\EnumClo_{S_{12}}$.
\end{proposition}
\begin{proof}
 Let $\phi = \displaystyle{\vee_{i \in [m]} C_i }$ where each $C_i$ is a conjunction of negative variables in $X =\{x_0,\dots,x_{n-1}\}$. We assume that no variable appear in all clauses otherwise we could do a trivial
 polynomial delay reduction to get rid of it.
 To a clause $C_i$ we associate the vector $v^i$ such that $v^i_j = 0$ if and only if $x_j \in C_i$.
 Let  $e^i$ denote  the vector which has the coefficient $e^i_i =1$ and its other coefficients are zero.
 We define $\ccS = \{ v^i \mid i \in [m]\} \cup \{ e^i \mid i \in [n]\}$. Note that there is natural bijection
 between assignments of $\phi$ and vectors of size $n$. Let $\phi(X)$ stand for the set of satisfying assignments of $\phi$.
 Note that $\phi(X) = \displaystyle{\cup_{i\in[m]} C_i(X)}$. And by definition $C_i(X)$ is equal to the set of all assignments where the variables in $C_i$ are equal to $0$. Since both $S_{10}$ and $S_{12}$ contain $\wedge$, $C_i(X) = \Cl_{S_{10}}(\{v^i\} \cup \{e^j \mid j\in[n]\})_{\1_{v^i}} = \Cl_{S_{12}}(\{v^i\} \cup \{e^j \mid j\in[n]\})_{\1_{v^i}}$. Since $C_i(X)$ contains $\Cl_{S_{10}}(\ccS_{\1_{v^i}})$ and $\Cl_{S_{12}}(\ccS_{\1_{v^i}})$ it is equal to these sets. Since no variable appear in all clauses, $\Cl_{S_{10}}(\ccS_{\1_{e^i}})$ and $\Cl_{S_{12}}(\ccS_{\1_{e^i}})$ are in $\phi(X)$ for all $i$. Hence, using the structure of $\Cl_{S_{10}}(\ccS)$ and $\Cl_{S_{12}}(\ccS)$ given in Lemma~\ref{lemma:struc}, we obtain $\Cl_{S_{10}}(\ccS) = \Cl_{S_{12}}(\ccS) = \phi(X)$.
\end{proof}

Since the reduction transforms an instance of $\enumSat{MONDNF}$ with $m$ clauses and $n$ variables into an instance of $\EnumClo_{S_{10}}$ or $\EnumClo_{S_{12}}$ with $m+n$ vectors of size $n$ any algorithm improving the delay of $\EnumClo_{S_{10}}$ or $\EnumClo_{S_{12}}$ would improve the one of $\enumSat{MONDNF}$.

\subsection{Majority and threshold}\label{sec:threshold}

An operation $f$ is a \emph{near unanimity} of arity $k$ if it satisfies $f(x_1,x_2,\dots,x_k) = x$ for each $k$-tuple
with at most one element different from $x$. The \emph{threshold} function of arity $k$, denoted by $Th^{k}_{k-1}$, is defined by $Th^{k}_{k-1}(x_1,\dots,x_k)$ is equal to $1$ if and only if at least $k-1$ of the elements $x_1,\dots,x_k$ are equal to one.
It is the smallest near unanimity operation over the booleans.
We use the Baker-Pixley theorem from universal algebra to characterize the closure by any clone which contains a threshold function by its projections of fixed size.

\subsubsection*{Majority}

The threshold function $Th^{3}_2$ is the majority operation over three booleans that we denote by $maj$ and the clone it generates is $D_2$. We could directly apply the Baker-Pixley theorem in this case, but we felt that for a computer science readership, it would be useful to present a proof of a special case: the characterization of $\Cl_{D_2}(\ccS)$. We also show how to obtain the best possible enumeration algorithm for deciding $\EnumClo_{D_2}$.

\begin{lemma}\label{lemma:maj}
Let $\ccS$ be set of vectors. Then $v$ belongs to $\Cl_{D_2}(\ccS)$ if and only if for all $i,j\in [n]$, $i\neq j$, there exists $x\in \ccS$ such that $x_{i,j}=v_{i,j}$.
\end{lemma}

\begin{proof}

    \noindent ($\Longrightarrow$)
     Given $a,b\in \{0,1\}$ and $i,j\in [n]$, $i\neq j$, we first show that if for all $v\in\ccS$, $v_i\neq a$ or $v_j\neq b$ then for all $u\in \Cl_{D_2}(\ccS)$, $u_i\neq a$ or $u_j\neq b$. It is sufficient to prove that this property is preserved by applying $maj$ to a vector set i.e. that if $\ccS$ has this property, then $maj(\ccS)$ has also this property. Let $x,y,z\in \ccS$, $v:=maj(x,y,z)$, and assume for contradiction that $v_{i,j}=(a,b)$. Since $v_i=a$, there is at least two vectors among $\{x,y,z\}$ that are equal to $a$ at index $i$. Without loss of generality, let $x$ and $y$ be these two vectors. Since for all $u\in\ccS$, $u_i\neq a$ or $u_j\neq b$, we have $x_j\neq b$ and $y_j\neq b$ and then $v_j\neq b$ which contradicts the assumption. We conclude that if $v\in \Cl_{D_2}(\ccS)$, then for all $i,j\in [n]$, there exists $u\in \ccS$ with $v_{i,j}=u_{i,j}$.

\noindent ($\Longleftarrow$)
  Let $k\leq n$ and let $a_1,...,a_k\in \{0,1\}$.
We show by induction on $k$, that if for all $i,j\leq k$ there exists $v\in\ccS$ with $v_i=a_i$ and $v_j=a_j$, then there exists $u\in \Cl_{D_2}(\ccS)$ with $u_1=a_1$, $u_2=a_2$, $...$, $u_k=a_k$. The assertion is true for $k=2$. Assume
it is true for $k-1$, and let $a_1,...,a_k\in \{0,1\}$. By induction hypothesis there
exists a vector $w\in \Cl_{D_2}(\ccS)$ with $w_1=a_1$, $...$, $w_{k-1}=a_{k-1}$.
By hypothesis, for all $i\leq k$ there exists $v^{i}\in \ccS$ with
$v^{i}_{i}=a_i$ and $v^{i}_{k}=a_k$. We then construct a sequence of vectors
$(u^i)_{i\leq k}$ as follow. We let $u^{1}=v^{1}$ and for all $1<i<k$,
$u^{i}=maj(w,u^{i-1},v^{i})$. We claim that $u:=u^{k-1}$ has the property sought
i.e. for all $i\leq k$, $u_{i}=a_{i}$. First let prove that for all $i<k$ and for
all $j\leq i$, $u^{i}_{j}=a_j$. It is true for $u^{1}$ by definition.
Assume now that the property holds for $u^{i-1}$, $i<k$. Then, by construction,
for all $j\leq i-1$, we have $u^{i}_j=a_j$ since $w_j=a_j$
and $u^{i-1}_{j}=a_j$. Furthermore, we have
$u^{i}_{i}=maj(w_i,u^{i-1}_{i},v^{i}_i)=a_i$ since $w_i=a_i$ and $v_i=a_i$. We
conclude that for all $i\leq k-1$, $u_i=u^{k-1}_{i}=a_i$.

    We claim now that for all $i<k$,  $u^{i}_k=a_k$. It is true for $u^{1}$. Assume it is true for $u^{i-1}$, $i<k$. Then we have $u^{i}_k=maj(w_k,u^{i-1}_{k},v^{i}_k)$ which is equal to $a_k$ since $u^{i-1}_{k}=a_k$ by induction and $v^{i}_k=a_k$ by definition. We then have $u_i=a_i$ for all $i\leq k$ which concludes the proof.
\end{proof}

As an immediate consequence we get the following corollary and proposition.

\begin{corollary}
$\Mem_{D_2}\in \P$.
\end{corollary}
\begin{proof}
Using Lemma~\ref{lemma:maj}, one decides whether a vector $v$ is in $\Cl_{D_2}(S)$, by checking for every pair of indices
$i,j$ whether there is a vector $w \in S$ such that $v_{i,j} = w_{i,j}$. \end{proof}

The complexity of the algorithm of the previous corollary is $O(mn^2)$, hence applying Prop~\ref{prop:partialsolutions}, we get an enumeration algorithm with delay $O(mn^3)$. We explain how to improve this delay in the next proposition.

\begin{proposition}\label{prop:speedup}
 $\EnumClo_{D_2}$ can be solved with delay $O(n^2)$.
\end{proposition}
\begin{proof}
We do a backtrack search and we explain how to efficiently decide $\Mem_{D_2}$ during the enumeration.
 We first precompute for each pair $(i,j)$ all values $(a,b)$ such that there exists $v\in \ccS$,
 $v_{i,j} = (a,b)$. When we want to decide whether the vector $v$ of size $l$ can be extended into a solution,
 it is enough that it satisfies the condition of Lemma~\ref{lemma:maj}. Moreover, we already know that
 $v_{[l-1]}$ satisfies the condition of Lemma~\ref{lemma:maj}. Hence we only have to check that the values of $v_{i,l}$ for all $i <l$
 can be found in $\ccS_{i,l}$ which can be done in time $O(l)$. The delay is the sum of the complexity of deciding $\Mem_{D_2}$ for each partial solution in a branch: $O(n^2)$.
\end{proof}

We can yet improve the delay by using an algorithm which does not use the backtrack search of Prop~\ref{prop:partialsolutions}.
The closure of $\ccS$ can be represented by a boolean formula $\phi_{\ccS}$ which describes the obstructions to being in $\Cl_{D_2}(\ccS)$. Its variables are $X = \{x_0,\dots,x_{n-1}\}$, therefore its assignments are naturally in bijection with the boolean vectors of size $n$ and we do not distinguish both. Recall that $\enumSat{2SAT}$ is the problem of generating all satisfying assignments of a $2CNF$ formula, that is a formula in conjunctive normal form where each clause is of size at most two.

\begin{proposition}\label{prop:2SAT}
 $\EnumClo_{D_2}$ can be can be solved with delay $O(n)$.
\end{proposition}

\begin{proof}
We prove that $\EnumClo_{D_2}$ can be reduced to $\enumSat{2SAT}$ by a polynomial delay reduction.
To an instance  $\ccS$, we associate a formula $\phi_{\ccS}$ in conjunctive normal form. Given $i\neq j$ and $v_{i,j}$ which is \emph{not} in $\ccS_{i,j}$, we define the clause $l_i \vee l_j$ where  $l_k$, for $k=i,j$, is $x_k$ if $v_k = 0$ and $\neg x_k$ otherwise. The formula $\phi_{\ccS}$ is the conjunction of all such clauses. By construction, the formula $\phi_{\Cl_{D_2}(\ccS)}$ is a $2CNF$.

 Notice that each clause forbids the assignment to take a pair of values which is not present in $\ccS$.
 Therefore a boolean vector $v$ seen as an assignment satisfies $\phi_{\ccS}$ if and only if every projection of size $2$ of $v$ is in $\ccS$. By Lemma~\ref{lemma:maj}, it means that $v \in \Cl_{D_2}(\ccS)$ if and only if $v$ satisfies   $\phi_{\ccS}$.

The proposition is proved since the problem $\enumSat{2SAT}$ can be solved with delay $O(n)$, where $n$ is the number of variables using an algorithm of Feder~\cite{feder1994network}.
%
 \end{proof}

We remark that we can apply the idea of Prop~\ref{prop:2SAT} to the clones of Subsection~\ref{sec:all} and to $D_1$ which all contain the $maj$ function. Indeed, closures by these clones are all characterized by their projections of size two, which allows to do a reduction to a $2CNF$ formula. The algorithms we obtain have delay $O(n)$. However, we felt it was relevant to deal with clones of Subsection~\ref{sec:all} separately, since the algorithms presented in Section~\ref{sec:all} are different and simpler. In particular, for the clones $R_0$ and $R$, the delay is only $O(1)$ if we allow a compact representation of the solutions.

\subsubsection*{The Baker-Pixley theorem}

We now state the Baker-Pixley theorem which generalizes Lemma~\ref{lemma:maj} to all clones which contain near unanimity terms. For more context on this theorem and universal algebra see \cite{burris2006course}.

\begin{theorem}[Baker-Pixley, adapted from~\cite{baker1975polynomial}]\label{thm:BP}
Let $\cF$ be a clone which contains a near unanimity term of arity $k$,
then $v \in \Cl_{\cF}(\ccS)$ if and only if for each set of indices $I$ of size $k-1$,
$v_I \in \Cl_{\cF}(\ccS)_I$.
\end{theorem}

This allows to settle the case of the two infinite families of clones
of our restricted lattice $S_{10}^k = <Th_k^{k+1}> $ and $S_{12}^k = <Th_k^{k+1},x \wedge (y \rightarrow z) > $.

\begin{corollary}\label{cor:coro}
If a clone $\cF$ contains $Th_k^{k+1}$ then $\Mem_{\cF}$ is solvable in time $O(mn^{k})$.
\end{corollary}

\begin{proof}
 Let $\ccS$ bet a set of vectors and let $v$ be a vector. By Theorem~\ref{thm:BP}, $v\ \in \Cl_{\cF}(\ccS)$
 if and only if for all $I$ of size $k$, $v_I \in \Cl_{\cF}(\ccS)_I$. First notice that $\Cl_{\cF}(\ccS)_I = \Cl_{\cF}(\ccS_I)$ because the functions of $\cF$ act coefficient-wise on $\ccS$. The algorithm generates for each $I$ of size $k$ the set  $\Cl_{\cF}(\ccS_I)$. For a given $I$, we first need to build the set $\ccS_I$ in time $m$ and then the generation of $\Cl_{\cF}(\ccS_I)$ can be done in constant time. Indeed, we can apply the classical incremental algorithm to generate the elements in $\Cl_{\cF}(\ccS_I)$, and the cardinal of $\Cl_{\cF}(\ccS_I)$  only depends on $k$ which is a constant. The time to generate all $\Cl_{\cF}(\ccS_I)$ is $O(mn^k)$ and then all the tests can be done in $O(n^k)$.
 \end{proof}

We have proved that the complexity of any closure problem in one of our infinite families is polynomial.
We can use the method of Proposition~\ref{prop:speedup} to obtain an algorithm with delay $O(n^{k})$
to enumerate the elements of a set closed by a clone containing a near unanimity function of arity $k+1$. When the arity is $3$, we can use the method of Proposition~\ref{prop:2SAT} to obtain an algorithm with the better delay $O(n)$.

The complexity of $\Mem_\cF$ is increasing with the smallest arity of a near unanimity function in $\cF$. Hence we now investigate the complexity of the \emph{uniform problem} when the clone is given as input. Let \textsc{ClosureTreshold} be the following problem: given a set $\ccS$ of vectors and an integer $k$ decide whether the vector $\mathbf{1}\in \Cl_{S_{10}^k}(\ccS) $.
It is a restricted version of the uniform problem, but it is already hard to solve because we can reduce the problem \textsc{HittingSet} to its complement.

\begin{theorem}\label{thm:coNP-completeness}
 \textsc{ClosureTreshold} is $\coNP$-complete.
\end{theorem}

\begin{proof}
First notice that the problem is in $\coNP$ since by Theorem~\ref{thm:BP}, the answer to the problem is negative if and only if one can exhibit a subset of indices $I$ of size $k$ such that no element of $\Cl_{S_{10}^k}(\ccS_I)$ is equal to $\one$.
Let us show that the later holds if and only if no element of $\ccS_I$ is equal to $\one$. We assume that $k\geq 3$ hence $S_{10}^k = <Th_k^{k+1}>$.
If no element in $\ccS_I $ is equal to $\one$, then the application
of $Th_k^{k+1}$ to $\ccS_I$ preserves this property. Indeed, consider $Th_k^{k+1}(v^1,\dots,v^{k+1})$. Each $v^i$ has a zero coefficient and since there are $k+1$ such vectors and the vectors are of size $k$, by the pigeonhole principle, there are $i,j,l$ such that $v^i_l = v^j_l = 0$. This implies that $Th_k^{k+1}(v^1,\dots,v^{k+1})\neq \one$.

\medskip

Let us show that \textsc{HittingSet} can be reduced to \textsc{ClosureTreshold}. Given a hypergraph $\cH=(V,\mathcal{E})$ and an integer $k$, \textsc{HittingSet} asks whether there exists
a subset $X\subseteq V$ of size $k$ that intersects all the hyperedges of $\cH$.
This problem is a classical NP-complete problem~\cite{garey2002computers}. Let $\cH=(V,\cE)$ be a hypergraph and
$k$ be an integer. Let $\bar{\cH}$ be the hypergraph on $V$ whose hyperedges are the complement of the hyperedges of $\cH$, and let $\ccS$ be the set of characteristic vectors of the hyperedges of $\bar{\cH}$. Then $\cH$ has a transversal of size $k$ if and
only if there is a set $I$ of indices of size $k$ such for all $v\in\ccS_I $,
$v \neq \one$. Indeed, $I$ is a hitting set of $\cH$ if for all $E\in \cE$, there exists $i\in I$ such that $i\in E$ which implies that $i\notin \overline{E}$ and then the characteristic vector $v$ of $\overline{E}$ is such that $v_i=0$.

Since the other direction is straightforward, we have proved that there is a set $I$ of indices of size $k$ such that for all $v\in\ccS_I $, $v \neq \one$ if and only if there is a set $I$ of indices of size $k$ such that for all $v\in\Cl_{S_{10}^k}(\ccS_I)$,
$v \neq \one$. By Theorem~\ref{thm:BP}, the later property is equivalent to $ \one \notin \Cl_{S_{10}^k}(\ccS)$.
Therefore we have given a polynomial time reduction from \textsc{HittingSet} to the complement of \textsc{ClosureTreshold} which proves the proposition.
\end{proof}

In fact, the result is even stronger. We cannot hope to get an FPT algorithm for \textsc{ClosureTreshold} parametrized by $k$
since \textsc{HittingSet} parametrized by the size of the hitting set is $\W[2]$-complete~\cite{flum2006parameterized}.
It implies that if we want to significantly improve the delay of our enumeration algorithm for the clone $S_{10}^k$, we should drop the backtrack search since it relies on solving $\Mem_{S_{10}^k}$.

\section{Larger Domains}\label{sec:largerdomains}

In this section, we try to extend some results of the boolean domain to larger domains, which is the simplest
generalization possible. The main problem is that, for domain of size $3$ or more, the lattice of clones is uncountable and not as well understood as Post's lattice. We show how we can do a backtrack search for threshold functions and algebraic functions
on larger domains. We also prove that for domain of size $3$, there is a clone generated by a single binary function $f$ such that $\Mem_{<f>}$ is $\NP$-hard.

\subsection{Tractable cases}

We show here how the cases of near unanimity term (subsection~\ref{sec:threshold} and \ref{sec:all}) and of the addition (subsection~\ref{sec:algebra}) can be generalized.
Using the Baker-Pixley theorem (Theorem~\ref{thm:BP}), we can get an equivalent to Corollary~\ref{cor:coro} and to Proposition~\ref{prop:speedup} in any domain size.

\begin{corollary}
If $\cF$ contains a near unanimity operation, then $\Mem_\cF \in P$.
\end{corollary}

\begin{proposition}
 If $\cF$ contains a near unanimity term of arity $k+1$, then $\EnumClo_\cF$ can be solved with delay $O(n^{k})$.
\end{proposition}

The second tractable case is a generalization of Subsection~\ref{sec:algebra}. The simplest generalization of the sum over $\mathbb{F}_2$ is to consider any finite group operation. Using Sim's stabilizer chain it is possible to decide membership in groups given by generators~\cite{furst1980polynomial,zweckinger2013computing}.

\begin{proposition}
 Let $f$ be a group operation over $D$. Then $\Mem_{<f>} \in P$.
\end{proposition}

The problem  $\Mem_{\cF}$ is known in the universal algebra community under the name \emph{subpower membership problem} or SMP. The method for groups also works for extensions of groups by multilinear operations such as ring or modules. For semigroups, some cases are known to be $\PSPACE$-complete and other are in $\P$~\cite{bulatov2016subpower,steindl2017subpower}. There are also works on SMP for clones which satisfy conditions from universal algebras: for expansions of finite nilpotent Maltsev algebras of prime power size or for residually finite algebras with cube terms, SMP is in $\P$~\cite{mayr2012subpower,szendrei2019subpower}, while it is hard for some algebras satisfying any strong Maltsev conditions~\cite{shriner2018hardness}.

\subsection{A limit to the backtrack search}

The last case we would like to extend is the clone generated by the conjunction.
A natural generalization is to fix an order on $D$ and to study the complexity
of $\Mem_{<f>}$ with $f$ monotone. Let $f$ be the function over $D=\{0,1,2\}$ defined by  $f(x,y) = \min(x+y,2)$. This function is monotone in each of its arguments. We prove that \textsc{Exact3Cover} reduces to $\Mem_{<f>}$, where \textsc{Exact3Cover} is the problem of finding an exact cover of a set by subsets of size $3$. The next proposition is from~\cite{mary2016efficient} and it has been independently proven using the same reduction for semigroups with idempotent elements in Lemma $5.2$ of~\cite{bulatov2016subpower}.

\begin{proposition}\label{prop:NPcover}
 $\Mem_{<f>}$ is $\NP$-complete.
\end{proposition}

\begin{proof}
 We reduce \textsc{Exact3Cover}, which is $\NP$-hard~\cite{garey2002computers}, to $\Mem_{<f>}$. Let $\ccS$ be an instance of \textsc{Exact3Cover}, that is a set of subsets of $[n]$ of size $3$. If we consider that the sets in $\ccS$ are represented by their characteristic vectors, then $\ccS$ is an instance of $\Mem_{<f>}$ and we now prove that $\one \in \Cl_{<f>}(\ccS)$ if and only if
 there is an exact cover of $\ccS$. First note that $f$ is associative,  therefore any element of $\Cl_{<f>}(\ccS)$ can be written $f(s^1,f(s^2,f(s^3, \dots)$ with $s^i \in \ccS$. Since $f$ is also commutative, the order of the operations is not relevant either. Hence we can represent any element  $v \in \Cl_{<f>}(\ccS)$ by a multiset of elements of $\ccS$, which combined by $f$ yields $v$.
 Notice that all elements in $\Cl_{<f>}(\ccS)$ can be seen as a multiset containing at most twice an element since $f(s,s) = f(s,f(s,s))$. If $s_i > 0$ then $f(s_i,s_i) =2$, therefore a multiset which yields $\1$ by $f$ should in fact be a set.
 Moreover a set which yields $\one$ satisfies that for all $i\leq n$ there is one and only one of its elements with a coefficient $1$ at the index $i$. Such a set is an exact cover of $\ccS$, which proves the reduction.

 The problem is in $\NP$ because an element $v$ is in $\Cl_{<f>}(\ccS)$, if and only if
 there is a multiset of elements of $\ccS$ such that applying $f$ to its elements yields $v$. We impose that the witnesses are multiset containing element at most twice. We have seen these witnesses always exist and it guarantees that they are of polynomial size.
\end{proof}

This hardness result implies that we cannot use the backtrack search to solve the associated enumeration algorithm.
However, if we allow a space proportional to the number of solutions, we can still get a polynomial delay algorithm for associative functions, a property satisfied by the function $f$ of the last proposition. Note that the space used is proportional to the \emph{number of solutions} while the backtrack search only requires a space polynomial in the input size.

\begin{proposition}
 If $f$ is an associative function of arity $2$, then $\EnumClo_{<f>} \in \DelayP$.
\end{proposition}
\begin{proof}
Let $\ccS$ be an instance of $\EnumClo_{<f>}$. Let $G$ be the directed graph with vertices $\Cl_{<f>}(\ccS)$
and from each $v \in \Cl_{<f>}(\ccS)$, there is an arc to $f(v,s)$ for all $s \in \ccS$.
Since $f$ is associative, by definition of $G$, every vertex of $\Cl_{<f>}(\ccS)$ is accessible from a vertex in $\ccS$.
Therefore we can do a depth-first traversal of the graph $G$ to enumerate all solutions.
A step of the traversal is in polynomial time: from an element $v$ we generate its neighborhood: $f(v,s)$ for $s \in \ccS$.
The computation of $f(v,s)$ is in time $O(n)$ and $|\ccS| = m$. We must also test whether the solution $f(v,s)$
has already been generated. This can be done in time $O(n)$ by maintaining a trie containing the generated solutions. In conclusion the delay of the enumeration algorithm is in $O(mn)$ thus polynomial.
\end{proof}

To obtain a polynomial space algorithm, we could try to use the \emph{reverse search} method~\cite{avis1996reverse}.
To do that, we want the graph $G$ to be a directed acyclic graph, which is the case if we require the function
to be monotone. The monotonicity also ensures that the depth of $G$ is at most $n(|D|-1)$.
However we also need to be able to compute for each element of $G$ a canonical ancestor
in polynomial time and it does not seem to be easy even when $f$ is monotone.

\begin{openproblem}
Find a property of $f$ which ensures the existence of an easy to compute ancestor so that we obtain a polynomial delay and polynomial space enumeration algorithm.
\end{openproblem}

\begin{openproblem}
 Generalize the previous algorithm to clones generated by several functions or functions of arity larger than three.
\end{openproblem}

\section{Non uniform operators}\label{sec:nonuniform}

In this section we generalize the class of operators acting on vectors of booleans. We relax the condition that an operator acts in the same way on each coefficient. On the other hand, we consider operators acting on $k$ solutions at most and in polynomial time, so that the saturation algorithm still works in incremental polynomial time.

We study the simplest case of \emph{unary} operators acting on a \emph{single} coefficient.
There are four functions from the boolean domain to the boolean domain.
When computing the closure, the identity has no effect, while the negation allows to have any value for a coefficient
which is also trivial to deal with. Therefore we are only interested in the two constant functions.

Let $\uparrow_i$ be the function which to a vector $v$ associates $\uparrow_i(v)$ such that for all $j\neq i$, $\uparrow_i(v)_j = v_j$ and $\uparrow_i(v)_i = 1$. We call this function an \emph{upward closure}.
The \emph{downward closures} are defined similarly: for each vector $v$ and index $j\neq i$, $\downarrow_i(v)_j = v_j$ and
$\downarrow_i(v)_i = 0$. We prove that we can add any sets of such functions to a set of operators over the boolean domain
and enumerate the elements of a closure with polynomial delay.

We generalize the problem $\Mem_{\cF}$ and $\EnumClo_{\cF}$, by adding to the input the set of upward and downward closures that can be used to compute the closure. Indeed, these operators work on a single coefficient and we should allow to grow their numbers when the size of the vectors increases to make our problem meaningfully different. In particular we want to allow the set of all downward closures as input to represent solutions closed by taking subsets.

For each set of boolean operators $\cF$, we define $\EnumUDClo_\cF$: given a set of vectors $\ccS$, a set of upward closures $\mathcal{U}$ and a set of downward closures $\mathcal{D}$, enumerate $\Cl_{\cF\cup \mathcal{U} \cup \mathcal{D}}(\ccS)$.
We present proofs that all problems $\EnumUDClo_\cF$ can be reduced to $\EnumClo_\cF$ in a tight way, which provides
good enumeration algorithms except for  $\EnumUDClo_\emptyset$. Note that when both $\uparrow_i \in \mathcal{U}$ and $\downarrow_i \in \mathcal{D}$, the value of the $i$th coefficient does not matter. Therefore we always assume that there is no $i$ such that $\uparrow_i \in \mathcal{U}$ and $\downarrow_i \in \mathcal{D}$.

\begin{proposition}\label{prop:and}
 $\EnumUDClo_{\emptyset}$ and $\EnumUDClo_{E_2}$ can be solved with delay $O(mn)$.
\end{proposition}
\begin{proof}
 The problem $\EnumUDClo_{\emptyset}$ can be reduced to $\enumSat{DNF}$ by mapping each element of $v \in \ccS$ to a clause $C_v$ of a formula $\phi_{\ccS}$. The variable $x_i$ is in $C_v$ if $v_i=1$, $\downarrow_i\notin D$ and $\neg x_i$ is in $C_v$ if $v_i=0$, $\uparrow_i\notin U$. By construction, the satisfying assignments of $\phi_{\ccS}$ are in bijection with $\Cl_{\mathcal{U},\mathcal{D}}(\ccS)$. Since the formula $\phi_{\ccS}$ has $m$ clauses of size at most $n$, one can enumerate its models with delay $O(mn)$.

Let $\ccS$ be a set of vectors, $\mathcal{U}$ be a set of upward closures and $\mathcal{D}$ be a set of downward closures.
 We want to decide whether $v \in \Cl_{E_2\cup \mathcal{U} \cup \mathcal{D}}(\ccS)$. If $v_i = 1$ and $\uparrow_i \in \mathcal{U}$ (resp. if $v_i = 0$ and $\downarrow_i \in \mathcal{D}$), then there is no constraint on the coefficient $i$ which can always be turned to the value $1$ (resp. $0$). Consider $i$ such that $v_i = 0$ and $\uparrow_i \in \mathcal{U}$. Assume that $v \in \Cl_{E_2\cup \mathcal{U} \cup \mathcal{D}}(\ccS)$, therefore there is a sequence of operations in $E_2\cup \mathcal{U} \cup \mathcal{D}$ which allows to build $v$ from $\ccS$. There are only two operators which can change the value of the coefficient $i$: $\uparrow_i$ and $\wedge$. Since $\wedge$ is decreasing and acts independently on each coefficient, we can remove all the uses of $\uparrow_i$ in the sequence of operations yielding $v$ without changing the final result. Similarly, if $v_i = 1$ it cannot be useful to use $\downarrow_i \in \mathcal{D}$. Therefore we can reduce the problem of deciding whether $v \in \Cl_{E_2\cup \mathcal{U} \cup \mathcal{D}}(\ccS)$ to the problem
 $\Mem_{E_2}$ on some projection of $\ccS$. A simple adaptation of Proposition~\ref{prop:unionfast} yields an algorithm with delay $O(mn)$.
 \end{proof}

 Note that it seems hard to improve the complexity of $\EnumUDClo_{\emptyset}$ since we can reduce $\enumSat{MONDNF}$ to it while preserving the size of the instance.  To each clause $C$ of a monotone DNF formula $\phi$ we map the vector $v_C$, such that $(v_C)_i =1$ if and only if $x_i$ is in the clause $C$. Let $\mathcal{U}$ be the set of all upward closures and $\ccS_{\phi}$ be the set of vectors $v_C$. The satisfying assignments of $\phi$ are in bijection with $\Cl_{ \mathcal{U}}(\ccS_{\phi})$ and $\ccS_{\phi}$
 has as many vectors as $\phi$ has clauses.

\begin{proposition}\label{prop:add}
$\EnumUDClo_{L_0}$ and $\EnumUDClo_{L_2}$ can be solved with delay $O(n)$.
\end{proposition}
\begin{proof}
Consider an index $i$, such that  $\uparrow_i \in \mathcal{U}$. Then if we consider $\uparrow_i(v + v)$, we obtain the base vector $e^{i}$ with $0$ everywhere except at position $i$. Therefore the value of coefficient $i$ is not constrained in  $\Cl_{L_0 \cup \mathcal{D} \cup \mathcal{U}(\ccS)}$ and we can solve the problem on $\ccS_{[n]\setminus \{i\}}$.

If $\downarrow_i \in \mathcal{D}$ and there is some vector $v$ with $v_i = 1$ then $ v + \downarrow_i(v)$ is the vector $e^{i}$.
If all vectors $v \in \ccS$ are such that $v_i = 0$ then all vectors $v' \in \Cl_{L_0 \cup \mathcal{U} \cup \mathcal{D}}(\ccS)$ have $v'_i = 0$. In both cases we can also solve the problem on $\ccS_{[n]\setminus \{i\}}$.
Therefore to solve the problem $\EnumUDClo_{L_0}$ it is enough to solve $\EnumClo_{L_0}$ on a projection of $\ccS$.

Since $\Cl_{L_2}(\ccS)$ is an affine space, it is equal to $v + \Cl_{L_0}(\ccS)$, therefore we can use the same argument to
solve $\EnumUDClo_{L_2}$ by reduction to $\EnumClo_{L_2}$.
\end{proof}

\begin{proposition} \label{prop:BP}
 Let $\cF$ be a clone with $Th_k^{k+1} \in \cF$ then $\EnumUDClo_{\cF}$ can be solved with delay $O(n^k)$.
\end{proposition}
\begin{proof}
First notice that by definition of the $\Cl$ operations, $ \Cl_{\cF \cup \mathcal{U} \cup \mathcal{D}}(\ccS) = \Cl_\cF(\Cl_{\cF \cup \mathcal{U} \cup \mathcal{D}}(\ccS))$.
Therefore we can apply Theorem~\ref{thm:BP} to the right hand side of the equation, since $\cF$ contains $Th_k^{k+1}$. As a consequence we have $ \Cl_{\cF \cup \mathcal{U} \cup \mathcal{D}}(\ccS) = \{ v \mid \forall |I| =k, \, v_I \in \Cl_{\cF}(\Cl_{\cF \cup \mathcal{U} \cup \mathcal{D}}(\ccS_I)) \}$. Applying the first equality again, we have $ \Cl_{\cF \cup \mathcal{U} \cup \mathcal{D}}(\ccS) = \{ v \mid \forall |I| =k, \, v_I \in \Cl_{\cF \cup \mathcal{U} \cup \mathcal{D}}(\ccS_I) \}$.
The characterization we obtain extends the one of Theorem \ref{thm:BP}. Since $\Cl_{\cF \cup \mathcal{U} \cup \mathcal{D}}(\ccS_I)$ can be generated in polynomial time for all $I$ of size $k$, we can use the same algorithm as in Corollary~\ref{cor:coro}.
\end{proof}

\begin{proposition} \label{prop:limits}
$\EnumUDClo_{S_{10}}$ and $\EnumUDClo_{S_{12}}$ can be solved with delay $O(mn)$.
\end{proposition}
\begin{proof}
We consider a set of upward closures $\mathcal{U}$ and we let $s'$ be the vector $s$ to which each element of
 $\mathcal{U}$ has been applied.
 Let $e^i$ be the vector such that $e^i_i = 1$ and $e^i_j = 0$ for $i\neq j$. We let $\ccS' = \ccS \cup \{s' \mid s \in \ccS\} \cup \{ e^i \mid \uparrow_i \in \mathcal{U}\}$.
 Let us prove that $\Cl_{S_{10}\cup\mathcal{U}}(\ccS) = \Cl_{S_{10}}(\ccS')$.
Note that the vector $\zero$ is in  $\Cl_{S_{10}}(\ccS)$ since $\wedge$ is in the clone $S_{10}$ and for every index $i$ there are two vectors $v,w$ such that $v_i =0$ and $w_i = 1$.
 The inclusion from right to left is clear: it is easy to build $s'$ by definition and $e^i = \uparrow_i (\zero)$.
 The other inclusion is true since $\uparrow_i (x) = x' \wedge (x \vee e^i)$.

We can prove by induction, using the monotonicity of $\vee$ and $\wedge$ that
$\Cl_{S_{10} \cup \mathcal{U} \cup \mathcal{D}}(\ccS)= \Cl_{\mathcal{D}}( \Cl_{\mathcal{U}}(\Cl_{S_{10}}(\ccS))$.
Therefore we only need to characterize $ \Cl_{\mathcal{D}}(\Cl_{S_{10}}(\ccS'))$, since we have shown that $\Cl_{\mathcal{U}}(\Cl_{S_{10}}(\ccS)) = \Cl_{S_{10}}(\ccS')$.
Let us define $y^i$ as the intersection of all elements $s \in \ccS'$ with $s_i =1$
to which are applied all $\downarrow_j \in \mathcal{D}$ for $j\neq i$.
Let us define $\ccS'' = \ccS' \cup \{ y^i \mid i \in [n] \}$, we now prove that
$ \Cl_{\mathcal{D}}(\Cl_{S_{10}}(\ccS')) = \Cl_{S_{10}}(\ccS'')$. The inclusion from right to left is clear: it is easy to build
$y^i$ since $\wedge$ is in the clone $S_{10}$. Now consider $v \in \Cl_{S_{10}}(\ccS')$ and $\downarrow_i \in \mathcal{D}$. By Lemma~\ref{lemma:struc}, $v$ can be seen as the union of some atoms of $\ccS$ intersected with $v$. Therefore, if we consider the union of $y^j$ with $j \in \1(v)$ and $j \neq i$, it is equal to $\downarrow_i(v)$ once intersected with $v$ since $y^j_i =0$ by construction.
It proves that $ \downarrow_i(v) \in \Cl_{S_{10}}(\ccS'')$ which proves the equality of the two closures.

As a conclusion, we have reduced an instance of $\EnumUDClo_{S_{10}}$ with $m$ vectors of size $n$ to an instance of $\EnumClo_{S_{10}}$ with at most $O(m+n)$ vectors of size $n$ which proves the result, using the algorithm of Proposition~\ref{prop:S10S12}.

 We now consider $\Cl_{S_{12} \cup \mathcal{U} \cup \mathcal{D}}(\ccS)$.
 Let $\uparrow_i \in \mathcal{U}$, we have assumed that there is $s \in \ccS$ with $s_i = 0$.
 Consider $x \in \Cl_{S_{12} \cup \mathcal{U} \cup \mathcal{D}}(\ccS)$, we have $x \wedge (\uparrow_i(s) \to s) = \downarrow_i(x)$. Therefore when we have an upper closure $\uparrow_i$, we can simulate the corresponding downward closure  $\downarrow_i(x)$ and thus forget about the coefficient which is not constrained.
 The proof is the same for downward closures, therefore we can reduce $\EnumUDClo_{S_{12}}$ to $\EnumClo_{S_{12}}$ on a projection of the instance, which by Proposition~\ref{prop:S10S12} yields the desired algorithm.
\end{proof}

If we allow more general operators as inputs, that is any $k$-ary function from $k$-tuples of vectors to vectors,
it is possible to use the incremental delay saturation algorithm. However, the problem of deciding membership in the closure becomes hard. For instance, we consider unary operators $O_{i,j,k}$, which act on the coefficients $i,j,k$ of a vector:
if $v_i=v_j=v_k=0$ then $O_{i,j,k}(v)_i=O_{i,j,k}(v)_j=O_{i,j,k}(v)_k=1$ otherwise the coefficients are left unchanged.
It is easy to reduce the problem \textsc{Exact3Cover} to deciding membership in the closure by these operators (a subset of them are given as input). To each subset $\{i,j,k\}$ is associated the operator  $O_{i,j,k}$, and the vector $\one$ belongs to the closure of $\ccS = \{ \zero \}$ by these operators if and only if the set of subsets has an exact cover.

\begin{openproblem}
 Is the membership problem hard for unary operators acting on $2$ coefficients ?
\end{openproblem}

There is another natural view on downward and upward closures. The question is to enumerate the subsets or supersets of
$\Cl_{\cF}(\ccS)$ for any fixed $\cF$. In that way, upward and downward closures do not mix with the operators in $\cF$.
We let the reader check that these problems are still solvable with polynomial delay by inspecting every case of the reduced Post's lattice. When generating a set of elements closed by subsets or supersets, we are often only interested in the minimal or maximal elements for inclusion. The next section deal with this natural extension.

\section{Minimal and maximal elements}\label{sec:max}

In this section, we are interested in enumerating minimal (resp. maximal) elements of $\Cl_{\cF}(\ccS)$ for a given family $\cF$. Given two vectors $v,u \in \{0,1\}^{n}$, we say that $v$ is smaller than $u$ if $\1(v) \subseteq \1(u)$, that is the vectors are ordered by the inclusion of the corresponding sets. A vector $v\in \Cl_\cF(\ccS)$ is said to be minimal (resp. maximal), if it is smaller (resp. larger) or incomparable with all vectors of $\Cl_\cF(\ccS)$. We denote by $\EnumCloMax_{\cF}$ and $\EnumCloMin_{\cF}$ the problem of, given $\ccS$, enumerating respectively the maximal and the minimal elements of $\Cl_{\cF}(\ccS) \setminus \{\zero,\one\}$. We exclude $\zero$ and $\one$ to obtain non trivial problems when $\Cl_{\cF}(\ccS)$ contains them.
Note that $\EnumCloMax_{\cF}$ and $\EnumCloMin_{\cF}$ cannot be solved by a simple saturation algorithm, therefore they are not naturally in $\IncP$ contrarily to all the problems studied in the previous sections.

\subsection{Cases with few solutions}

There are many clones which have only a polynomial number of minimal or maximal elements which can be easily found,
making their enumeration trivial.

The maximal elements of a closure by $E_2$ are the original vectors since the operation is decreasing. The minimal elements are the intersection of all base vectors containing $i$ for each $i \in [n]$.
The closures by the clones $M_2$, $R_0$ and $R$ are characterized as union of atoms or the intersection with a base vector of a union of atoms for $S_{10}$ and $S_{12}$.
As a consequence the minimal elements are these atoms which can be computed in polynomial time.
For the maximal elements, if we assume that the union of all atoms is equal to $\one$ then it is the unions of all atoms but one.

The hierarchies $S_{10}^k$ and $S_{12}^k$ contain respectively $S_{10}$ and $S_{12}$ and thus a closure by  $S_{10}^k$ and $S_{12}^k$
contains the atoms defined for $S_{10}$ and $S_{12}$. They are their minimal elements, which are thus easy to produce. However since the structure of $S_{10}^k$ and $S_{12}^k$ is slightly different from $S_{10}$ and $S_{12}$, they have many more maximal elements and we will see in Section~\ref{ssec:infinite} that they are much harder to enumerate.

We are left with the cases of minimal and maximal elements when the clone contains the majority function, maximal elements when the clone contains a threshold function of arity larger than $3$ and minimal and maximal elements for the algebraic clones $L_0$ and $L_2$.

\subsection{Majority }\label{ssec:}

For the five clones which contain the majority operator, we can use the reduction to a $2CNF$ formula of Proposition~\ref{prop:2SAT}. Then we can generate the maximal (resp. minimal) elements of their closure by generating maximal (resp. minimal) models of a $2CNF$ formula. Note that the minimal models of a $2CNF$ formula are the negations of the maximal models of the same formula where each variable is negated. Therefore the problem of enumerating maximal models is equivalent to the problem of generating minimal models.

The problem of enumerating the maximal models of a $2CNF$ formula has already been studied and can be solved with polynomial delay as explained in~\cite{kavvadias2000generating}.
The algorithm works by eliminating by resolution all positive variables in the formula. It means that for a given variable
$x$, and for all clauses $A \vee x$ and $B \vee \neg x$ all clauses $A \vee B$ are created and $A \vee x$, $B \vee \neg x$ are removed.
This operation does not change the set of maximal models. This step can be successively applied in time $O(n^3)$ to obtain $m'$ clauses where $n$ is the number of variables.
Then the formula is a $2CNF$, such that each variable is always positive or always negative.
When enumerating maximal models of such a formula all positive variables are set to $1$ and the formula is thus an antimonotone $2CNF$.
Therefore the problem can be reduced to the enumeration of maximal independent sets of the graph whose vertices are variables and the edges are clauses. The later problem can be solved with delay $O(n^3)$ and space $O(n^2)$ since $m' \leq n^2$, see~\cite{LawlerLK80}. Note that any graph can be encoded into an antimonotone $2CNF$, therefore improving the enumeration of maximal solutions of antimonotone $2CNF$ means finding a better enumeration algorithm for maximal independent sets in general graphs which is a hard open problem.
%

\subsection{Linear algebra}

The clone $L_0$ is generated by $+$, hence the closure by $L_0$ is a vector space. Therefore we want to understand the problem of finding the minimal and maximal vectors of a vector space over $\mathbb{F}_2$. Since the minimum of a vector space is always $\zero$,  we consider the minimal \emph{non zero vectors}, a concept which is related to circuits of a matroid.
A binary matroid $\mathcal{M}$ is represented by a set of vectors over $\mathbb{F}_2$.
 A circuit of $\mathcal{M}$ is a minimal --for the inclusion-- dependent set, that is a subset $S$ of vectors in $\mathcal{M}$ such that there is a linear combination of elements in $S$ equal to $\zero$.

\begin{theorem}
 The problem $\EnumCloMin_{L_0}$ is equivalent to the enumeration of the circuits of a binary matroid.
\end{theorem}
\begin{proof}
Let $\mathcal{M}$ be a binary matroid and let $M$ be the matrix whose columns are the vectors of $\mathcal{M}$.
 The dependent sets of $\mathcal{M}$ are the solutions to $Mx = 0$ or supersets of these elements.
 We can find a basis of the kernel of $M$ in polynomial time and we call it $B$. Notice that by definition,
 $\Cl_{L_0}(B)$ contains only dependent sets of $\mathcal{M}$ and amongst them all its minimal dependent sets.
 Hence solving $\EnumCloMin_{L_0}$ on $B$ is equivalent to the enumeration of the circuits of $\mathcal{M}$.

 Now assume that we are given a set $\mathcal{S}$ of binary vectors. We can compute a basis $B$ of the vector space $\Cl_{L_0}(\mathcal{S})$ and then we compute in polynomial time a matrix $M$ such that $B$ is a basis of the kernel of $M$.
 The binary matroid defined by the columns vectors of $M$ has for circuits the minimal elements of $\Cl_{L_0}(B)= \Cl_{L_0}(\mathcal{S})$, which proves the reduction.
\end{proof}

The enumeration of the circuits of a matroid can be done in incremental polynomial time~\cite{khachiyan2005complexity}, but it is an open question to find a polynomial delay for this problem even when the matroid is binary. In particular, it is not possible to enumerate the circuits by increasing Hamming weight, since finding the one with smallest Hamming weight is $\NP$-hard~\cite{berlekamp1978inherent}. Moreover, we cannot use the backtrack search, since deciding whether a set of vectors can be extended into a circuit is $\NP$-complete~\cite{durand2003inference}. Therefore we have an incremental polynomial time algorithm for $\EnumCloMin_{L_0}$ but improving it would solve an important open question.

The reduction done for $\EnumCloMin_{L_0}$ also works for $\EnumCloMin_{L_2}$, but since it is an affine space,
the problem is equivalent to finding the minimal solutions of a non homogeneous system. As a consequence solving $\EnumCloMin_{L_2}$ is equivalent to enumerating the circuits of a matroid containing \emph{a given element}. This later problem is also in $\IncP$~\cite{khachiyan2005complexity}.

\begin{openproblem}
 Can we relate $\EnumCloMax_{L_0}$ to a known enumeration problem as we have done for $\EnumCloMin_{L_0}$?
 In particular, is $\EnumCloMax_{L_0} \in \IncP$.
\end{openproblem}

\subsection{The infinite part}\label{ssec:infinite}

We now study $\EnumCloMax_{S_{10}^k}$ and $\EnumCloMax_{S_{12}^k}$ and we show that they are both equivalent to the enumeration of maximal independent sets of a hypergraph of dimension $k$, a problem denoted by $\EnumKMaxIndSet$. The dimension of a hypergraph is the size of its largest hyperedge. An independent set of a hypergraph is a subset of its vertices that does not contain any hyperedge. It is said to be maximal, if it is maximal by inclusion. The problem of enumerating the maximal independent sets or equivalently the minimal transversal of a hypergraph of dimension bounded by $k$ belongs to $\IncP$ (c.f. \cite{eiter1995identifying,boros98,boros2000efficient}) but no polynomial delay algorithm is known.  Therefore the equivalence of
$\EnumKMaxIndSet$ with $\EnumCloMax_{S_{10}^k}$ and $\EnumCloMax_{S_{12}^k}$ proves that they are in $\IncP$ since this class is closed under polynomial delay reduction.


We first give a characterization of the vectors of $Cl_{S_{10}^k}(\ccS)$ and $Cl_{S_{12}^k}(\ccS)$. By Theorem \ref{thm:BP}, a vector $v$ belongs to $Cl_{S_{10}^k}(\ccS)$ (resp. $Cl_{S_{12}^k}(\ccS)$) if and only if for every subset of indices $I$ of size $k$, $v_I\in Cl_{S_{10}^k}(\ccS)_{I}$ (resp. $v_I\in Cl_{S_{12}^k}(\ccS)_{I}$). So the vectors of the closure are completely determined by the closures of the projections on the $k$-subsets of indices of $\ccS$.  Since the projection of the closure is equal to the closure of the projection, the following two lemmas give a characterization of the vectors in $Cl_{S_{10}^k}(\ccS)_{I}$ and $Cl_{S_{12}^k}(\ccS)_{I}$.

\begin{lemma}\label{lm:S10}
     Let $k>3$, $\ccS $ be a set of vectors each of size $k$ and let $v$ be a vector of size $k$. Then $v\in Cl_{S_{10}^k}(\ccS)$ if and only if:
     \begin{itemize}
          \item There exists $u\in \ccS$ such that $\1(v)\subseteq \1(u)$
          \item For all $i,j\leq k$ such that $v_i=1$ and $v_j=0$, there exists $u \in \ccS$ with $u_i=1$ and $u_j=0$
     \end{itemize}
\end{lemma}

\begin{proof}
     ($\Rightarrow$) Let us prove that if for a given set of indices $I$, $|I|\leq k$, $\ccS$ has the property that no projection of its vectors on $I$ is equal to $\one$, then this property is preserved by the application of $Th_k^{k+1}$ and thus $\Cl_{S_{10}^k}(\ccS)$ has the same property. This implies that if no vector $u$ of $\ccS$ is such that $\1(v)\subseteq \1(u)$, then $v\notin \Cl_{S_{10}^k}(\ccS)$.
     Let us consider $Th_k^{k+1}(u^1,\dots,u^{k+1})$, each $u^i$ has at least one zero on indices in $I$. Since there are $k+1$ such vectors and $I$ is of size at most $k$, by the pigeonhole principle, there are $i,j\leq k$ and $l\in I$ such that $u^i_l = u^j_l = 0$. This implies that $Th_k^{k+1}(u^1,\dots,u^{k+1})\neq \one$. Thus, if $v\in Cl_{S_{10}^k}(\ccS)$, there exists $u\in \ccS$ with $\1(v)\subseteq \1(u)$.

     Let us prove now that if, for all vectors $u \in \ccS$, we have $u_{i,j}\neq (1,0)$, then the same is true in $\Cl_{S_{10}^k}(\ccS)$, since the property is preserved by the application of $Th_k^{k+1}$. Assume that $v=Th_k^{k+1}(u^1,\dots,u^{k+1})$ with $v_i=1$. Then at least $k$ vectors from  $u^1,\dots,u^{k+1}$ are equal to $1$ on the $i^{\text{th}}$ coordinate. By hypothesis, all these $k$ vectors are equal to $1$ on the $j^{\text{th}}$ coordinate which implies that $v_j=1$.

    \medskip

     \noindent($\Leftarrow$)
     Let $v$ be a vector and assume that there exists $u\in \ccS$ such that $\1(v)\subseteq \1(u)$ and assume that for all $i,j\leq k$ such that $v_i=1$ and $v_j=0$, there exists $u \in \ccS$ with $u_i=1$ and $u_j=0$. We show that $v\in Cl_{S_{10}^k}(\ccS)$. By rearranging the coordinates, assume without loss of generality that $v_i=1$ for all $i\leq \ell$ and $v_i=0$ for all $\ell < i \leq k$. We show by induction that for all $\ell\leq r \leq k$, $v_{[r]}\in Cl_{S_{10}^k}(\ccS)_{[r]}$.
     For $r=\ell$, we have $v_{[r]}=\one$. By assumption there exists $u\in \ccS$ with $\1(v)\subseteq 1(u)$, so we have $v_{[r]}=u_{[r]}\in \ccS_{[r]}\subseteq Cl_{S_{10}^k}(\ccS)_{[r]}$. Assume now that $v_{[r-1]} \in Cl_{S_{10}^k}(\ccS)_{[r-1]}$ for a given $r$ with $ \ell <r \leq k$ and let us show that $v_{[r]} \in Cl_{S_{10}^k}(\ccS)_{[r]}$. Since $v_{[r-1]} \in Cl_{S_{10}^k}(\ccS)_{[r-1]}$ there exists $v'\in Cl_{S_{10}^k}(\ccS)$ such that $v'_{[r-1]}=v_{[r-1]}$. Observe first that if $v'_r=0$ then the result directly holds since $v_r=0$ and then $v'_{[r]}=v_{[r]}$. So assume that $v'_r=1$. Since by assumption, for each $i,j\leq k$ such that $v_i=1$ and $v_j=0$ there exists a vector $u \in \ccS$ with $u_i=1$ and $u_j=0$, we have that for each $i\leq \ell$, there exists $u^{i}\in \ccS$ with $u^i_i=1$ and $u^i_r=0$. We now construct a sequence of vectors $(z^{i})_{i\leq \ell}$ that belong to $Cl_{S_{10}^k}(\ccS)$ and we prove by induction that  each $z^{i}$ is such that :
     \begin{itemize}
          \item    $z^{i}_j=1$ for all $j\leq i$
          \item $z^{i}_j=0$ for all $\ell<j\leq r$
      \end{itemize}
      Then we have that $z^{\ell}_{[r]}=v_{[r]}$ which will prove that $v_{[r]}\in Cl_{S_{10}^k}(\ccS)_{[r]}$.
      We first let $z^{1}=Th_k^{k+1}(v',v',...,v',u^1,u^1)$. Clearly $z^{1}_1=1$ since $u^{1}=1$ and $v'_1=1$. Furthermore, since $u^{1}_r=0$, $z^{1}_r=0$  and since $v'_{j}=0$ for all $\ell< j <r$, $z^{1}_j=0$ for all $\ell< j <r$.
      Now for $1<i\leq \ell$ we let $z^{i}=Th_k^{k+1}(v',v',...,v',u^i,z^{i-1})$. Let us prove that $z^{i}$ has the sought properties. Let $j\leq i$. Notice that since $v'_j=1$, $z^{i}_j=1$ if either $u^{i}_j$ or $z^{i-1}_j$ is equal to one. If $j=i$ then $u^{i}_{j}=1$ and then $z^{i}_j=1$. Now if $j<i$ then by induction, $z^{i-1}_j=1$ and then $z^{i}_j=1$. Now let $\ell<j\leq r$. If $j<r$, we have that $v'_j=0$ and then $z^{i}_j=0$. Now since by induction $z^{i-1}_r=0$ and since by definition $u^{i}_r=0$, we have $z^{i}_r=0$.
      So for every $i\leq \ell$, $z^{i}$ has the sought property which implies in particular that $z^{\ell}_{[r]}=v_{[r]}$. Since at each step, $z^{i}$ is constructed by applying the threshold operator on vectors of $Cl_{S_{10}^k}(\ccS)$, $z^{\ell}$ belongs to $Cl_{S_{10}^k}(\ccS)$ and thus $v_{[r]}\in Cl_{S_{10}^k}(\ccS)_{[r]}$.
\end{proof}

\begin{lemma}
     Let $k>3$, $\ccS $ be a set of vectors each of size $k$ and let $v$ be a vector of size $k$. Assume that for all $i < k$, there is $s \in \ccS$ such that $s_i = 0$. Then $v\in Cl_{S_{12}^k}(\ccS)$ if and only if:
     \begin{itemize}
	  \item There exists $u\in \ccS$ such that $\1(v)\subseteq \1(u)$
          \item For all $i,j\leq k$ such that $v_i\neq v_j$, there exists $u \in \ccS$ with $u_i\neq u_j$
     \end{itemize}
\end{lemma}
\begin{proof}
      ($\Rightarrow$) Assume that $v\in Cl_{S_{12}^k}(\ccS)$.
      Observe first that if for two coordinates $i,j$ we have $u_i=u_j$ for all vectors in $\ccS$, then applying any boolean operator produce a vector for which both coordinates are equal. So if $v_i\neq v_j$ there must be a vector in $\ccS$ with $u_i\neq u_j$. Now let us show that there exists a vector $u\in \ccS$ with $\1(v)\subseteq \1(u)$.
      By Lemma \ref{lm:S10} if no vector $u$ of $\ccS$  is such that $\1(v)\subseteq \1(u)$, then $v\notin Cl_{S_{10}^k}(\ccS)$. Furthermore the extra operator $x \wedge (y \to z)$ which is in $S_{12}^k$ is decreasing because of the $\wedge$. So if no vector $u$ of $\ccS$  is such that $\1(v)\subseteq \1(u)$, then $v\notin Cl_{S_{12}^k}(\ccS)$.

      \medskip

      \noindent($\Leftarrow$) Assume that $v$  has the following properties
      \begin{itemize}
          \item For all $i,j\leq k$ such that $v_i\neq v_j$, there exists $u \in \ccS$ with $u_i\neq u_j$
          \item There exists $s\in \ccS$ such that $\1(v)\subseteq \1(s)$
      \end{itemize}
      and let us show that $v\in Cl_{S_{12}^k}(\ccS)$. The first property shows it is enough to prove it projected on $\1(s)$.
      We can apply Lemma~\ref{lemma:struc} to  $Cl_{S_{12}}(\ccS_{\1(s)})$, which shows it is equal to $Cl_{BF}(\ccS_{\1(s)})$, since there is no index on which all elements of $\ccS$ are $1$. Since $S_{12}$ is included in $S_{12}^k$, $Cl_{S_{12}^k}(\ccS)$ contains $Cl_{BF}(\ccS_{\1(s)})$. Using the second property of $v$, we have that $v \in Cl_{BF}(\ccS_{\1(s)})$ which implies that
      $v\in Cl_{S_{12}^k}(\ccS)$.
\end{proof}

Notice that if two coordinates $i$ and $j$ are equal for every vector of $\ccS$, they will remain equal in the closure. So one can simplify the input vector set by deleting all but one equivalent coordinates. So assuming that no coordinates are equivalent, the two previous lemmas together with Theorem \ref{thm:BP} give the two following corollaries.

\begin{corollary}\label{cor:S10}
     Let $k>3$, $n\geq k$, $\ccS $ be a set of vectors each of size $n$ and let $v$ be a vector of size $n$. Then $v\in Cl_{S_{10}^k}(\ccS)$ if and only if:
     \begin{itemize}
          \item For all $i,j\leq n$ such that $v_i=1$ and $v_j=0$, there exists $u \in \ccS$ with $u_i=1$ and $u_j=0$
          \item For all $I\subseteq [n]$, $|I|=k$, there exists $u\in \ccS$ such that $\1(v_I)\subseteq \1(u_I)$
      \end{itemize}
\end{corollary}

\begin{corollary}\label{cor:S12}
     Let $k>3$, $n\geq k$, $\ccS $ be a set of vectors each of size $n$ and let $v$ be a vector of size $n$. Then $v\in Cl_{S_{12}^k}(\ccS)$ if and only if for all $I\subseteq [n]$, $|I|=k$, there exists $u\in \ccS$ such that $\1(v_I)\subseteq \1(u_I)$.
\end{corollary}

We are now ready to reduce $\EnumKMaxIndSet$ to $\CloMax_{S_{10}^k}$ and $\CloMax_{S_{12}^k}$.
Let $\cH=(C,\cE)$ be a hypergraph of dimension $k$ on $n$ vertices and assume
for simplicity that $\cH$ is $k$-regular (i.e. all its hyperedges are of size
exactly $k$). Let $\cT=\{\chi_{T} \mid T \in \binom{V}{k} \setminus \cE \}$
where $\binom{V}{k}$ denotes the subsets of $V$ of size $k$ and $\chi_{T}$
denotes the characteristic vector of $T$. Now for $i\leq n$, let $e^i$ be the
vector of size $n$ such that $e^i_i=1$ and $e^{i}_j=0$ for all $j\neq i$. If we denote by $\cU$ the set of vectors $\cT \cup \{e^i : i\leq n\}$ we have the following property.

\begin{lemma}
 $Cl_{S_{10}^k}(\cU)$ (resp. $Cl_{S_{12}^k}(\cU)$) is the set of the characteristic vectors of the independent sets of $\cH$.
\end{lemma}

\begin{proof}
Since for every $i,j\leq n$ there is $e^{i}\in \cU$ such that $e^{i}_i=1$ and $e^{i}_j=0$, by Corollary \ref{cor:S10} and \ref{cor:S12}, a vector belongs to $Cl_{S_{10}^k}(\cU)$ (resp. $Cl_{S_{12}^k}(\cU)$) if and only if for every subset of indices $I\subseteq [n]$ of size $k$, there exists $u\in \cU$ such that $\1(v_I)\subseteq \1(u_I)$. So $v\in Cl_{S_{10}^k}(\cU)$ (resp. $v\in Cl_{S_{12}^k}(\cU)$) if and only if for every subset $I$ of $\1(v)$ of size $k$ there exists $u\in \cU$ such that $v_I=u_I=\one$. Since every vectors of $\cU$ has at most $k$ ones, $v\in Cl_{S_{10}^k}(\cU)$ (resp. $v\in Cl_{S_{10}^k}(\cU)$) if and only for all subsets $I$ of $\1(v)$ of size $k$ there exists $u\in \cT$ such that $u$ is the characteristic vector of $\1(v_I)$. By construction of $\cT$, the later is  equivalent to the fact that $I$ is not a hyperedge of $\cH$. So, since $\cH$ is $k$-regular, $v\in Cl_{S_{10}^k}(\cU)$ (resp. $v\in Cl_{S_{12}^k}(\cU)$) if and only if $\1(v)$ contains no hyperedge of $\cH$, i.e. $v$ is the characteristic vector of an independent set of $\cH$.
\end{proof}

The bijections between the closures of $\cU$ and the independent sets of $\cH$ are also bijections between the maximal elements of the closures and the maximal independent sets. Since $\cU$ can be built in time polynomial in $\cH$, we have proved the existence of polynomial delay reductions as stated in the next corollary.

\begin{corollary}\label{cor:reductionindependent}
      There are polynomial delay reductions from $\EnumKMaxIndSet$ to $\EnumCloMax_{S_{10}^k}$ and $\EnumCloMax_{S_{12}^k}$.
\end{corollary}


Now we show that $\EnumCloMax_{S_{10}^k}$ and $\EnumCloMax_{S_{12}^k}$ can be
reduced to $\EnumKMaxIndSet$, proving that these three problems are equivalent.
Let $\ccS$ be a set of binary vectors of size $n$, and $k<n$ be an integer. We denote by $\cH_k(\ccS)$ the hypergraph on $n$ vertices such that the hyperedges are the subsets $I\subseteq [n]$, $|I|\leq k$ such that $u_I\neq \one$ for all $u\in \ccS$. A hyperedge of $\cH_k(\ccS)$ is a subset of coordinates that are never set to one at the same time in $\ccS$.
We denote by $I(\cH_{k(\ccS)})$ the set of independent sets of $\cH_{k}$ and we denote by $Max(I(\cH_{k(\ccS)}))$ its maximal ones. For a family of sets $\cT$, let us denote by $\chi(\cT)$ the set of characteristic vectors of the sets in $\cT$.

We first show that $Cl_{S_{12}^k}(\ccS)$ is the set of characteristic vectors of $I(\cH_{k(\ccS)})$. Since the inclusion ordering between sets correspond to the classical ordering between their characteristic vectors, this implies that  $Max(Cl_{S_{12}^k}(\ccS)) = \chi(Max(I(\cH_{k(\ccS)})))$. On the other hand $Cl_{S_{10}^k}(\ccS) \neq  \chi(I(\cH_{k(\ccS)}))$, but we will see that $Max(Cl_{S_{10}^k}(\ccS))=Max(Cl_{S_{12}^k}(\ccS))=\chi(Max(I(\cH_{k(\ccS)})))$.

\begin{theorem}\label{thm:S12k-equiv}
$\EnumCloMax_{S_{12}^k}$ and $\EnumKMaxIndSet$ are equivalent.

\end{theorem}
\begin{proof}
We first prove that $Cl_{S_{12}^k}(\ccS) = \chi(I(\cH_{k(\ccS)}))$.

     We have $T\in \chi(I(\cH_{k(\ccS)}))$\\
     $\Leftrightarrow$ for each subset $I\subseteq \1(v)$, $|I|\leq k$, $I$ is not a hyperedge of $\cH_{k}(\ccS)$\\
     $\Leftrightarrow$ for each subset $I\subseteq \1(v)$, $|I|\leq k$, there exists $u\in \ccS$ such that $u_{I}=\one$\\
     $\Leftrightarrow$ $v\in Cl_{S_{12}^k}(\ccS)$ (By Corollary \ref{cor:S12})

     Given $\ccS$, $\cH_k(\ccS)$ can be constructed in time $O(mn^{k})$ where $m$ is the number of vectors in $\ccS$.
     Moreover we have proved that $Cl_{S_{12}^k}(\ccS) = \chi(I(\cH_{k(\ccS)}))$ and $\chi$ can be computed in linear time,
     which proves that this is a polynomial delay reduction from $\EnumCloMax_{S_{12}^k}$ to $\EnumKMaxIndSet$.
     Since Corollary~\ref{cor:reductionindependent} proves the reduction in the other direction, the two problems are equivalent.
      \end{proof}

  We now prove that $\EnumCloMax_{S_{10}^k}$ and $\EnumKMaxIndSet$ are equivalent by proving that the maximal elements of  $Cl_{S_{12}^k}(\ccS)$ and $Cl_{S_{10}^k}(\ccS)$ are the same.

\begin{theorem}
$\EnumCloMax_{S_{10}^k}$ and $\EnumKMaxIndSet$ are equivalent.
 \end{theorem}

\begin{proof}
Since $S_{10}^k \subseteq S_{12}^k$ we have that $Cl_{S_{10}^k}(\ccS)\subseteq Cl_{S_{12}^k}(\ccS)$.

 \textbf{Claim 1}:$Max(Cl_{S_{12}^k}(\ccS)) \subseteq Max(Cl_{S_{10}^k}(\ccS))$

Let $v\in Max(Cl_{S_{12}^k}(\ccS))$ and let us show that $v\in Cl_{S_{10}^k}(\ccS)$. Assume it is not, by corollaries  \ref{cor:S10} and \ref{cor:S12} and since $v\in Cl_{S_{12}^k}(\ccS)$, we know that there exists $i,j\leq n$ such that $v_{i,j}=(1,0)$ while the vector $(1,0)$ does not belong to $\ccS_{i,j}$. By maximality of $v$ in $Cl_{S_{12}^k}(\ccS)$, the vector $v'$ obtained from $v$ by setting $j$ to $1$ does not belong to $Cl_{S_{12}^k}(\ccS)$, so by corollary \ref{cor:S12}, there exists $I\subseteq \1(v')$, $|I|\leq k$ such that $v'_{I}\notin \ccS_{I}$. Notice that $j$ must belong to $I$ since otherwise $v$ would not belong to $Cl_{S_{12}^k}(\ccS)$. Now let $I':= I\setminus \{j\} \cup \{i\}$. We have that $I'\subseteq \1(v)$ and $|I'|\leq k$. So again, by corollary \ref{cor:S12} there exists $u\in \ccS$ such that $u_{I'}=v_{I'}$. Since $u_{I}\neq v'_{I}$ and $v'_j=1$ we have that $u_j=0$. But since $u_i=v_i=1$, $u_{i,j}=(1,0)$ contradicting the fact that the vector $(1,0)$ does not belong to $\ccS_{i,j}$. So $v\in Cl_{S_{10}^k}(\ccS)$, and since $v$ is maximal in $Cl_{S_{12}^k}(\ccS)$ it is also maximal in $Cl_{S_{10}^k}(\ccS)$.

Assume that $v\in Max(Cl_{S_{10}^k}(\ccS))$. Since  $Cl_{S_{10}^k}(\ccS)\subseteq Cl_{S_{12}^k}(\ccS)$, $v\in Cl_{S_{12}^k}(\ccS)$. Now assume that $v$ is not maximal in $Cl_{S_{12}^k}(\ccS)$, then there exists $v'\in Max(Cl_{S_{12}^k}(\ccS))$ such that $v<v'$. But then, by Claim 1, $v'\in Cl_{S_{10}^k}(\ccS)$ which contradicts the maximality of $v$ in $Cl_{S_{10}^k}(\ccS)$.

We have thus proved that $Max(Cl_{S_{12}^k}(\ccS)) = Max(Cl_{S_{10}^k}(\ccS))$ which proves that
$\EnumCloMax_{S_{10}^k}$ and $\EnumCloMax_{S_{12}^k}$ are equivalent. Then by Theorem~\ref{thm:S12k-equiv} and the transitivity of equivalence, we have that  $\EnumCloMax_{S_{10}^k}$ and $\EnumKMaxIndSet$ are equivalent.
\end{proof}









\acknowledgements{
The authors have been partially supported by the French Agence Nationale de la Recherche,
AGGREG project reference ANR-14-CE25-0017-01 and we thank the members of the project and Mamadou Kant\'e for interesting discussions about enumeration. We also thank Florent Madelaine for his help with CSP and universal algebra. We thank Bruno Zanuttini for his remarks about representing closures as boolean formulas.
}

\bibliographystyle{plain}
 \bibliography{biblio.bib}

\begin{thebibliography}{10}

\bibitem{avis1996reverse}
David Avis and Komei Fukuda.
\newblock Reverse search for enumeration.
\newblock {\em Discrete Applied Mathematics}, 65(1-3):21--46, 1996.

\bibitem{baker1975polynomial}
Kirby~A. Baker and Alden~F. Pixley.
\newblock Polynomial interpolation and the chinese remainder theorem for
  algebraic systems.
\newblock {\em Mathematische Zeitschrift}, 143(2):165--174, 1975.

\bibitem{berlekamp1978inherent}
E.R. Berlekamp, R.J. McEliece, and H.C.A. Van~Tilborg.
\newblock {On the inherent intractability of certain coding problems}.
\newblock {\em IEEE Transactions on Information Theory}, 24(3):384--386, 1978.

\bibitem{bohler2002boolean}
Elmar B{\"o}hler and Heribert Vollmer.
\newblock Boolean functions and post’s lattice with applications to
  complexity theory.
\newblock {\em Lecture Notes for Logic and Interaction}, 2002.

\bibitem{boros2000efficient}
Endre Boros, Vladimir. Gurvich, Khaled Elbassioni, and Leonid Khachiyan.
\newblock An efficient incremental algorithm for generating all maximal
  independent sets in hypergraphs of bounded dimension.
\newblock {\em Parallel Processing Letters}, 10(04):253--266, 2000.

\bibitem{boros98}
Endre Boros, Vladimir Gurvich, and Peter~L. Hammer.
\newblock Dual subimplicants of positive boolean functions.
\newblock {\em Optimization Methods and Software}, 10:147--156, 1998.

\bibitem{bulatov2012enumerating}
Andrei Bulatov, V{\'\i}ctor Dalmau, Martin Grohe, and D{\'a}niel Marx.
\newblock Enumerating homomorphisms.
\newblock {\em Journal of Computer and System Sciences}, 78(2):638--650, 2012.

\bibitem{bulatov2016subpower}
Andrei Bulatov, Marcin Kozik, Peter Mayr, and Markus Steindl.
\newblock The subpower membership problem for semigroups.
\newblock {\em International Journal of Algebra and Computation},
  26(07):1435--1451, 2016.

\bibitem{burris2006course}
Stanley Burris and Hanamantagida~Pandappa Sankappanavar.
\newblock {\em A Course in Universal Algebra}, volume~78 of {\em Graduate Texts
  in Mathematics}.
\newblock Springer-Verlag, New York-Berlin, 1981 (2012 online version).

\bibitem{capelli2016structural}
Florent Capelli.
\newblock {\em Structural restrictions of CNF-formulas: applications to model
  counting and knowledge compilation}.
\newblock PhD thesis, Universit{\'e} Paris Diderot, 2016.

\bibitem{capelli2018incremental}
Florent Capelli and Yann Strozecki.
\newblock Incremental delay enumeration: Space and time.
\newblock {\em Discrete Applied Mathematics}, 2018.

\bibitem{cook1971complexity}
Stephen~A. Cook.
\newblock The complexity of theorem-proving procedures.
\newblock In {\em Proceedings of the third annual ACM symposium on Theory of
  computing}, pages 151--158. ACM, 1971.

\bibitem{creignou1997generating}
Nadia Creignou and J-J. H{\'e}brard.
\newblock On generating all solutions of generalized satisfiability problems.
\newblock {\em RAIRO-Theoretical Informatics and Applications}, 31(6):499--511,
  1997.

\bibitem{durand2003inference}
Arnaud Durand and Miki Hermann.
\newblock The inference problem for propositional circumscription of affine
  formulas is conp-complete.
\newblock In {\em STACS}, pages 451--462. Springer, 2003.

\bibitem{DurandS11}
Arnaud Durand and Yann Strozecki.
\newblock Enumeration complexity of logical query problems with second-order
  variables.
\newblock In {\em Proceedings of the 20th Conference on Computer Science
  Logic}, pages 189--202, 2011.

\bibitem{eiter1995identifying}
Thomas Eiter and Georg Gottlob.
\newblock Identifying the minimal transversals of a hypergraph and related
  problems.
\newblock {\em SIAM Journal on Computing}, 24(6):1278--1304, 1995.

\bibitem{elbassioni2018enumerating}
Khaled Elbassioni and Kazuhisa Makino.
\newblock Enumerating vertices of 0/1-polyhedra associated with 0/1-totally
  unimodular matrices.
\newblock In {\em 16th Scandinavian Symposium and Workshops on Algorithm Theory
  (SWAT 2018)}. Schloss Dagstuhl-Leibniz-Zentrum fuer Informatik, 2018.

\bibitem{feder1994network}
Tom{\'a}s Feder.
\newblock Network flow and 2-satisfiability.
\newblock {\em Algorithmica}, 11(3):291--319, 1994.

\bibitem{flum2006parameterized}
J{\"o}rg Flum and Martin Grohe.
\newblock {\em Parameterized complexity theory}.
\newblock Springer Science \& Business Media, 2006.

\bibitem{furst1980polynomial}
Merrick Furst, John Hopcroft, and Eugene Luks.
\newblock Polynomial-time algorithms for permutation groups.
\newblock In {\em Foundations of Computer Science, 1980., 21st Annual Symposium
  on}, pages 36--41. IEEE, 1980.

\bibitem{garey2002computers}
Michael~R. Garey and David~S. Johnson.
\newblock {\em Computers and intractability}, volume~29.
\newblock W Freeman New York, 2002.

\bibitem{JohnsonP88}
David~S. Johnson, Christos~H. Papadimitriou, and Mihalis Yannakakis.
\newblock On generating all maximal independent sets.
\newblock {\em Information Processing Letters}, 27(3):119--123, 1988.

\bibitem{kavvadias2000generating}
Dimitris~J. Kavvadias, Martha Sideri, and Elias~C. Stavropoulos.
\newblock Generating all maximal models of a boolean expression.
\newblock {\em Information Processing Letters}, 74(3-4):157--162, 2000.

\bibitem{khachiyan2005complexity}
Leonid Khachiyan, Endre Boros, Khaled Elbassioni, Vladimir Gurvich, and
  Kazuhisa Makino.
\newblock On the complexity of some enumeration problems for matroids.
\newblock {\em SIAM Journal on Discrete Mathematics}, 19(4):966--984, 2005.

\bibitem{knuth2011combinatorial}
Donald~E. Knuth.
\newblock Combinatorial {A}lgorithms, part 1, volume 4a of {T}he {A}rt of
  {C}omputer {P}rogramming, 2011.

\bibitem{LawlerLK80}
Eugene~L. Lawler, Jan~Karel Lenstra, and A.~H. G.~Rinnooy Kan.
\newblock Generating all maximal independent sets: Np-hardness and
  polynomial-time algorithms.
\newblock {\em SIAM J. Comput.}, 9(3):558--565, 1980.

\bibitem{mary2013enumeration}
Arnaud Mary.
\newblock {\em {\'E}num{\'e}ration des Dominants Minimaux d’un graphe}.
\newblock PhD thesis, Universit{\'e} Blaise Pascal, 2013.

\bibitem{mary2016efficient}
Arnaud Mary and Yann Strozecki.
\newblock Efficient enumeration of solutions produced by closure operations.
\newblock In {\em 33rd Symposium on Theoretical Aspects of Computer Science,
  STACS 2016}, volume~47, 2016.

\bibitem{wepa2016}
Arnaud Mary and Yann Strozecki.
\newblock Generating maximal solutions given by closure operations.
\newblock Workshop on Enumeration Problems and Applications, 2016.

\bibitem{mayr2012subpower}
Peter Mayr.
\newblock The subpower membership problem for {M}al'cev algebras.
\newblock {\em International Journal of Algebra and Computation},
  22(07):1250075, 2012.

\bibitem{murakami2014efficient}
Keisuke Murakami and Takeaki Uno.
\newblock Efficient algorithms for dualizing large-scale hypergraphs.
\newblock {\em Discrete Applied Mathematics}, 170:83--94, 2014.

\bibitem{post1941two}
Emil~Leon Post.
\newblock {\em The two-valued iterative systems of mathematical logic}.
\newblock Princeton University Press, 1941.

\bibitem{read1975bounds}
Robert~C. Read.
\newblock Bounds on backtrack algorithms for listing cycles, paths, and
  spanning trees.
\newblock {\em Networks}, 5:237--252, 1975.

\bibitem{reith2003optimal}
Steffen Reith and Heribert Vollmer.
\newblock Optimal satisfiability for propositional calculi and constraint
  satisfaction problems.
\newblock {\em Information and Computation}, 186(1):1--19, 2003.

\bibitem{schaefer1978complexity}
Thomas~J. Schaefer.
\newblock The complexity of satisfiability problems.
\newblock In {\em Proceedings of the tenth annual ACM symposium on Theory of
  computing}, pages 216--226. ACM, 1978.

\bibitem{shriner2018hardness}
Jeff Shriner.
\newblock Hardness results for the subpower membership problem.
\newblock {\em International Journal of Algebra and Computation},
  28(05):719--732, 2018.

\bibitem{steindl2017subpower}
Markus Steindl.
\newblock The subpower membership problem for bands.
\newblock {\em Journal of Algebra}, 489:529--551, 2017.

\bibitem{phd_strozecki}
Yann Strozecki.
\newblock {\em Enumeration complexity and matroid decomposition}.
\newblock PhD thesis, Universit\'e Paris Diderot - Paris 7, 2010.

\bibitem{strozecki2013enumerating}
Yann Strozecki.
\newblock On enumerating monomials and other combinatorial structures by
  polynomial interpolation.
\newblock {\em Theory of Computing Systems}, 53(4):532--568, 2013.

\bibitem{szendrei2019subpower}
{\'A}gnes Szendrei, Peter Mayr, and Andrei Bulatov.
\newblock The subpower membership problem for finite algebras with cube terms.
\newblock {\em Logical Methods in Computer Science}, 15, 2019.

\bibitem{uno1997algorithms}
Takeaki Uno.
\newblock Algorithms for enumerating all perfect, maximum and maximal matchings
  in bipartite graphs.
\newblock {\em Algorithms and Computation}, pages 92--101, 1997.

\bibitem{williams2010subcubic}
Virginia~Vassilevska Williams and Ryan Williams.
\newblock Subcubic equivalences between path, matrix and triangle problems.
\newblock In {\em Foundations of Computer Science (FOCS), 2010 51st Annual IEEE
  Symposium on}, pages 645--654. IEEE, 2010.

\bibitem{zweckinger2013computing}
Stephan Zweckinger.
\newblock Computing in direct powers of expanded groups.
\newblock Master's thesis, 2013.

\end{thebibliography}

\end{document}